\documentclass[11pt,twoside]{article}
\pdfoutput=1

\usepackage{amsfonts}
\usepackage{amsmath}
\usepackage{amsthm}
\usepackage{amssymb}
\usepackage{stmaryrd}
\usepackage{amscd}
\usepackage{eucal}
\usepackage{dsfont}
\usepackage{bm}
\usepackage{graphicx}

\usepackage{a4}
\usepackage{times}
\usepackage{cite}
\usepackage{microtype}
\usepackage[title]{appendix}
\def\beq{\begin{equation}}
\def\eeq{\end{equation}}
\def\bea{\begin{eqnarray}}
\def\eea{\end{eqnarray}}

\def\beann{\begin{eqnarray*}}
\def\eeann{\end{eqnarray*}}

\let\a=\alpha \let\be=\beta \let\g=\gamma \let\de=\delta
\let\e=\varepsilon \let\z=\zeta  \let\th=\theta

\let\dh=\vartheta \let\k=\kappa \let\la=\lambda \let\m=\mu
\let\n=\nu \let\x=\xi \let\p=\pi \let\r=\rho \let\s=\sigma
\let\om=\omega \let\ps=\psi
\let\ph=\varphi \let\Ph=\phi \let\PH=\Phi \let\Ps=\Psi

 \let\Si=\Sigma 
\let\La=\Lambda \let\G=\Gamma \let\D=\Delta

\let\qd=\quad \let\qqd=\qquad 

\def\epp{\, .}
\def\epc{\, ,}

\def\tst#1{{\textstyle #1}}
\def\dst#1{{\displaystyle #1}}

\theoremstyle{plain}
\newtheorem{theorem}{Theorem}
\newtheorem*{theorem*}{Theorem}
\newtheorem{lemma}{Lemma}
\newtheorem*{lemma*}{Lemma}

\newtheorem*{proposition*}{Proposition}
\newtheorem{corollary}{Corollary}
\newtheorem*{corollary*}{Corollary}
\newtheorem{conjecture}{Conjecture}
\newtheorem*{conjecture*}{Conjecture}

\theoremstyle{definition}

\newtheorem*{remark}{Remark}
\newtheorem*{question*}{Question}

\def\2{\frac{1}{2}} \def\4{\frac{1}{4}}

\def\6{\partial}

\def\+{\dagger}

\def\<{\langle} \def\>{\rangle}
\def\llb{\llbracket} \def\rrb{\rrbracket}

\let\auf=\uparrow \let\ab=\downarrow

\def\i{{\rm i}}

\def\rd{{\rm d}}

\def\ctg{\, {\rm ctg}\,}

\DeclareMathOperator{\re}{e}

\DeclareMathOperator{\sh}{sh}

\DeclareMathOperator{\tr}{tr}

\DeclareMathOperator{\Int}{Int}
\DeclareMathOperator{\Ext}{Ext}

\DeclareMathOperator{\End}{End}
\DeclareMathOperator{\id}{id}
\DeclareMathOperator{\ad}{ad}
\DeclareMathOperator{\res}{res}
\DeclareMathOperator{\card}{card}

\renewcommand{\Re}{\operatorname{Re}}
\renewcommand{\Im}{\operatorname{Im}}

\def\bv{\mathbf{b}}
\def\cv{\mathbf{c}}

\def\nuv{{\boldsymbol{\nu}}}
\def\siv{\boldsymbol{\sigma}}
\def\xiv{\boldsymbol{\xi}}

\def\fa{\mathfrak{a}}

\def\fe{\mathfrak{e}}

\DeclareMathOperator{\fz}{\mathfrak{z}}

\def\ks{h_R}

\renewcommand{\ks}{\k}
\newcommand{\nex}{\ell}
\allowdisplaybreaks

\begin{document}

\thispagestyle{empty}

\begin{center}

{\Large \bf
Thermal form-factor expansion of the dynamical two-point
functions of local operators in integrable quantum chains}

\vspace{10mm}

{\large
Frank G\"{o}hmann,\textsuperscript{1)}\footnote[2]{e-mail: goehmann@uni-wuppertal.de}
Karol K. Kozlowski\textsuperscript{2)}\footnote[1]{e-mail: karol.kozlowski@ens-lyon.fr} and
Mikhail D. Minin\textsuperscript{1)}\footnote[3]{e-mail: minin@uni-wuppertal.de}
%and Junji Suzuki$^\ddagger$
}\\[3.5ex]
\textsuperscript{1)}Fakult\"at f\"ur Mathematik und Naturwissenschaften,\\
Bergische Universit\"at Wuppertal,
42097 Wuppertal, Germany\\[1.0ex]
\textsuperscript{2)}Univ Lyon, ENS de Lyon, Univ Claude Bernard,\\ CNRS,
Laboratoire de Physique, F-69342 Lyon, France\\[1.0ex]
%$^\ddagger$Department of Physics, Faculty of Science, Shizuoka University,\\
%Ohya 836, Suruga, Shizuoka, Japan

\vspace{50mm}

%\vspace{17mm}
%Resubmitted March 8, 2021 %\today
%{\it Dedicated to Professor Barry M. McCoy on the occasion of his 80th birthday}
%\vspace{19mm}

{\large {\bf Abstract}}

\end{center}

\begin{list}{}{\addtolength{\rightmargin}{9mm}
               \addtolength{\topsep}{-5mm}}
\item
Evaluating a lattice path integral in terms of spectral data and
matrix elements pertaining to a suitably defined quantum transfer matrix,
we derive form-factor series expansions for the dynamical two-point
functions of arbitrary local operators in fundamental Yang-Baxter
integrable lattice models at finite temperature. The summands in
the series are parameterised by solutions of the Bethe Ansatz
equations associated with the eigenvalue problem of the quantum
transfer matrix. We elaborate on the example of the XXZ chain for
which the solutions of the Bethe Ansatz equations are sufficiently
well understood in certain limiting cases. We work out in detail the
case of the spin-zero operators in the antiferromagnetic massive regime
at zero temperature. In this case the thermal form-factor series
turn into series of multiple integrals with fully explicit
integrands. These integrands factorize into an operator-dependent
part, determined by the so-called Fermionic basis, and a part
which we call the universal weight as it is the same for all
spin-zero operators. The universal weight can be inferred
from our previous work. The operator-dependent part is rather
simple for the most interesting short-range operators. It is
determined by two functions $\r$ and $\om$ for which we
obtain explicit expressions in the considered case. As an
application we rederive the known explicit form-factor series
for the two-point function of the magnetization operator and
obtain analogous expressions for the magnetic current and the
energy operators.
%\\[2ex]
%{\it PACS: 05.30.-d, 75.10.Pq}
\end{list}

\clearpage

\section{Introduction}
The main objective of this work is to develop a more general picture
of the thermal form-factor approach \cite{DGK13a} to dynamical
two-point functions of fundamental Yang-Baxter integrable lattice
models \cite{GKKKS17}. For the specific example of the XXZ spin chain
we further explore its connection with the Fermionic basis introduced
by H. Boos et al.~\cite{BJMST08a,JMS08,BJMS09a} and obtain
explicit form-factor series for the dynamical two-point functions of
some spin-zero operators in the antiferromagnetic massive regime in
the low-temperature limit.

Form factor expansions have been successfully used in order
to study integrable quantum field theories and lattice models.
Still, there is a considerable difference in attitude towards
the two types of theories, and the methods that have been employed
to study them are partially complementary. In the quantum field theory
setting form factors are solutions of a set of functional equations.
The entirety of all possible solutions fixes the matrix elements
of all operators that can be defined in the theory and determines
the theory in this sense. For this reason the functional
equations are also called form-factor axioms. After a number
of explorative studies by various authors \cite{KaWe78,Smirnov86}
they were formulated in their full glory by Kirillov and
Smirnov \cite{KiSm87} in 1987 (see also \cite{Smirnov92,BFKZ99}).
The axioms determine the form factors directly in infinite volume.
The calculation of finite-size corrections to the form factors
from the axioms requires an indirect and somewhat intricate
reasoning (see e.g.~\cite{BLSV23pp} and literature cited therein).

The integrable lattice models can be interpreted as models of
many-body quantum mechanics. Unlike the integrable quantum field
theories they have spaces of states and spaces of local operators
that can be constructed explicitly from tensor products of
finite dimensional vector spaces. Their integrability is typically
related to the existence of families of commuting transfer
matrices of underlying vertex-models that commute with their
Hamiltonians and can be diagonalized by means of the algebraic
Bethe Ansatz. Accordingly, form factors can be directly calculated
for finite-size systems in a basis of Bethe Ansatz eigenstates
\cite{Slavnov89,KMT99a}.  The result are expressions for form
factors parameterised by solutions of the Bethe Ansatz equations.
These have turned out to be useful for the numerical calculation
of correlation functions \cite{CaMa05,SST04}, but it is hard to
obtain sufficiently explicit expressions in the infinite volume limit 
\cite{Slavnov90,IKMT99,KKMST09b, KKMST11a,DGKS15a,Kozlowski17,KiKu19}.
Nevertheless, the form-factor series based on the cited works
were useful for the analysis of the long-time, large-distance
asymptotics of the two-point functions of the simplest local
operators in the XXZ chain \cite{KKMST11b,KKMST12,Kozlowski18,%
Kozlowski19}.

A method for the calculation of finite temperature dynamical
correlation functions of integrable lattice models directly in
the infinite volume limit was devised in \cite{GKKKS17}. It
utilizes the notion of a quantum transfer matrix that was
originally introduced \cite{Suzuki85} for the calculation of
the thermodynamic properties of quantum spin chains \cite{Kluemper93}
and later turned out to be useful for the calculation of their
finite temperature reduced density matrix \cite{GKS04a,GKS05}.
More recently it was observed \cite{DGKS15b,FGK23pp} that the
solution sets of the Bethe Ansatz equations pertaining to the
excited states of the quantum transfer matrix of the XXZ chain
in a finite magnetic field and in the low-temperature limit
are particularly simple in that they do not involve string
patterns in the complex plane. For the model in the
antiferromagnetic massive regime, this was the precondition for
deriving a form-factor series for the dynamical two-point functions
of the local magnetization in which the $(n + 1)$th term comprises
the contribution of $2n$ elementary excitations in a $2n$-fold
integral over a form-factor density entirely expressed in terms of
known special functions \cite{BGKS21a,BGKSS21}. This is to
be contrasted with previous works in which the the form-%
factor densities were given in terms of multiple integrals
\cite{JiMi95}, in terms of multiple residues \cite{DGKS15a},
or in terms of Fredholm determinants \cite{DGKS16b}.

Integrable quantum field theories and integrable lattice models
are not independent of one another. Beyond the scope of the
axiomatic method integrable quantum field theories appear
as scaling limits in certain subsectors of the integrable
lattice models. This point of view is very fruitful. It may,
for instance, be utilized to study the field theory in finite
volume. In conjuction with the Fermionic basis approach,
which was originally developed for the XXZ chain, \textit{viz.},
on the lattice \cite{BJMST08a, JMS08,BJMS09a,JMS21}, it
lead to novel insights into the structure of the form factors
of the Sine-Gordon model \cite{JMS10,JMS11,JMS11c,JMS11b}.
On the other hand, the form-factor axioms may be interpreted
as a special case of the $q$-Knizhnik-Zamolodchikov
equations \cite{Smirnov92b} and as such are linked to the
representation theory of the quantum affine algebra
$U_q (\widehat{\mathfrak{sl}_2})$. Hence, it is not surprising 
that the form factors of the infinite XXZ chain in its
antiferromagnetic massive regime satisfy the form-factor
axioms \cite{JkMQ94}, a fact that has not received much attention
so far. Within the Fermionic basis approach a \emph{reduced}
q-Knizhnik-Zamolodchikov equation (rqKZ) \cite{BJMST04b,JMS21}
plays a prominent role in the characterization of a
generalized reduced density matrix of the XXZ chain. This
equation was generalized to the finite-temperature (finite
inhomogeneous vertex-model) setting in \cite{AuKl12}. We shall
see below that the latter formulation can be further generalized
to characterize properly normalized form factors on the lattice.

Our work is organized as follows. In Section~\ref{sec:two-point_fun}
we define a lattice path-integral representation of the dynamical
correlation functions of two arbitrary local operators in
fundamental Yang-Baxter integrable lattice models and express
it as a form-factor series in which every summand is defined entirely
in terms of Bethe ansatz data. Section~\ref{sec:spin_zero_xxz}
considers the XXZ chain as an example. We introduce properly
normalized form factors of spin-zero operators which we then
first characterize by their properties, most notably reduction
and a generalized version of the discrete rqKZ equation of
\cite{AuKl12}. Using reduction and exchange symmetry we then
obtain the form factors of the magnetization and of the magnetic 
current in terms of ratios of eigenvalues of the quantum transfer
matrix. For the form factors of general spin-zero operators we 
obtain multiple-integral representations of the same form
as for the generalized reduced density matrix \cite{BoGo09}.
We then show that the double integral factorizes and establish
the relation of the properly normalized form factors with the
Fermionic basis, in particular with the JMS theorem \cite{JMS08}.
In Section~\ref{sec:massive_low_T_xxz} we consider the zero-temperature
limit of the form-factor expansions of the two-point functions of
spin-zero operators for the XXZ chain in the antiferromagnetic
massive regime. In this case the normalization factor turns
into a `universal weight function' that can be inferred in
explicit form from previous work \cite{BGKS21a,BGKSS21}. The
remaining form-factor densities become explicit as well in
the simplest cases of the spin and spin current operators. We
present our conclusions in Section~\ref{sec:conclusions}. A few
technical details are deferred to a number of appendices.

\section{Two-point functions of elementary blocks}
\label{sec:two-point_fun}
Extending our paper \cite{GKKKS17} we would like to derive
in this section a form-factor series representation for the most
general dynamical two-point function of two local operators in
a fundamental Yang-Baxter integrable model.

\subsection{Fundamental lattice models}
Fundamental models can be interpreted as quantum chains with
local degrees of freedom in a finite dimensional Hilbert
space ${\cal H} = {\mathbb C}^d$, equipped with the canonical
Hermitian scalar product. We fix a basis
\begin{equation}
     \{e_\a\}_{\a = 1}^d \subset {\mathbb C}^d \epp
\end{equation}
It follows that the set $\{e_\a^\be\}_{\a, \be = 1}^d \subset
\End {\mathbb C}^d$, defined by
\begin{equation}
     e_\a^\be e_\g = \de_\g^\be e_\a
\end{equation}
for $\a, \be, \g = 1, \dots, d$, is a basis of $\End {\mathbb C}^d$.
In this basis the identity $I_d \in \End {\mathbb C}^d$ is
expanded as $I_d = e_\a^\a$. Here and in the following summation
over double Greek indices is implied.

The space of states of the models under consideration is the
tensor product space ${\cal H}_L = ({\mathbb C}^d)^{\otimes L}$.
The number of factors $L$ is called the number of lattice
sites or the length of the quantum chain.
%We shall consider chains of even length.
The embedding of the basis of elementary
endomorphisms $\{e_\a^\be\}_{\a, \be = 1}^d$ into $\End
({\mathbb C}^d)^{\otimes L}$ is defined by
\begin{equation}
     {e_j}_\a^\be = I_d^{\otimes (j - 1)} \otimes e_\a^\be
                    \otimes I_d^{\otimes (L - j)} \epc
\end{equation}
where $j = 1, \dots, L$. This definition allows us to introduce 
`$m$-site operators'. For every $A \in \End ({\mathbb C}^d)^{\otimes m}$,
$m \le L$, and $\{j_1, \dots, j_m\} \subset \{1, \dots, L\}$ we set
\begin{equation}
     A_{j_1, \dots, j_m} = A^{\a_1 \dots \a_m}_{\be_1 \dots \be_m}
                         {e_{j_1}}_{\a_1}^{\be_1} \dots {e_{j_m}}_{\a_m}^{\be_m} \epp
\end{equation}
We say that $A$ acts non-trivially only on sites $j_1, \dots, j_m$.

A fundamental integrable lattice model is defined in terms
of a finite dimensional solution $R: {\mathbb C}^2
\rightarrow \End {\mathbb C}^d \otimes {\mathbb C}^d$
of the Yang-Baxter equation
\begin{equation} \label{ybe}
     R_{1,2} (\la, \mu) R_{1,3} (\m, \n) R_{2,3} (\m, \n) =
	R_{2,3} (\m, \n) R_{1,3} (\m, \n) R_{1,2} (\la, \mu) \epp
\end{equation}
Solutions of the Yang-Baxter equation are called $R$-matrices. In
addition to the Yang-Baxter equation we require the $R$-matrix to
satisfy the following conditions,
\begin{subequations}
\label{fund}
\begin{align}
     & R(\la, \la) = P && \text{regularity} \epc \label{reg} \\[1ex]
     & R_{1, 2} (\la, \m) R_{2, 1} (\m, \la) = \id && \text{unitarity} \epc \\[1ex]
     & R^t (\la, \m) = R (\la, \m) && \text{symmetry} \epc \label{rsym}
\end{align}
\end{subequations}
where $P = e^\a_\be \otimes e^\be_\a$ is the transposition
matrix and $\id $ the identity.

In the following we shall often take implicit recourse to the graphical
representation
\begin{equation} \label{drawr}
     R^{\a \g}_{\be \de} (\la, \m) \: = \:
        \text{\raisebox{-48.5pt}{\includegraphics[width=.24\textwidth]{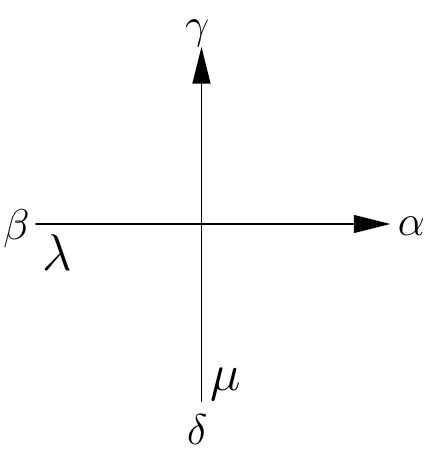}}}
\end{equation}
of the $R$-matrix, where $\a, \be, \g, \de = 1, \dots, d$
(for $d = 2$ also $\pm$ or $\auf, \ab$).

We define a monodromy matrix
\begin{equation} \label{monoperp}
     T_{\perp} (\la) = R_{0, L} (\la, 0) \dots R_{0,1} (\la, 0)
\end{equation}
and the associated row-to-row transfer matrix
\begin{equation}
     t_\perp (\la) = \tr_0 \{T_{\perp} (\la)\} \epp
\end{equation}
The corresponding nearest-neighbour quantum chain Hamiltonian is
\begin{equation} \label{fundham}
     H_0 = \frac{1}{\k} t_\perp' (0) t_\perp^{-1} (0) \epc
\end{equation}
where $\k \in {\mathbb C}$ can be chosen to our convenience, e.g.\
in such a way that $H_0$ is Hermitian. Along with $t_\perp$ we consider
$\overline{t}_\perp$ defined as
\begin{equation}
     \overline{t}_\perp (\la) = \tr_0 \{T_{\perp}^{-1} (\la)\}
        = \tr_0 \{R_{1, 0} (0, \la) \dots R_{L, 0} (0, \la)\} \epp
\end{equation}
It follows that
\begin{equation}
     \overline{t}_\perp (0) = t_\perp^{-1} (0) \epc \qd
     H_0 = - \frac{1}{\k} t_\perp (0) \overline{t}_\perp' (0) \epp
\end{equation}
For $N \in 2 {\mathbb N}$ let
\begin{equation} \label{ftnapprox}
     \r_{N, L} (\be) =
        \Bigl(t_\perp \bigl(- \tst{\frac{\be}{\k N}}\bigr)
	      \overline{t}_\perp \bigl(\tst{\frac{\be}{\k N}}\bigr)\Bigr)^\frac{N}{2}
	      \epp
\end{equation}
Then
\begin{equation} \label{rhoop}
     \lim_{N \rightarrow \infty} \r_{N, L} (\be) = \re^{- \be H_0} \epp
\end{equation}

For finite $N$ the product of transfer matrices $\r_{N, L} (\be)$
is an approximation to the operator $\re^{- \be H_0}$, which can
be interpreted as the statistical operator, if $1/\be = T$ is the
temperature, or as the time evolution operator, if $\be = \i t$,
where $t$ is the time. We shall call $N$ the Trotter number
and the limit $N \rightarrow \infty$ the Trotter limit. The
significance of (\ref{ftnapprox}), (\ref{rhoop}) is that these
formulae realize the operators that determine the time evolution
and the thermal average for a given Hamiltonian in terms of
the $R$-matrix that encodes its integrable structure.

\subsection{\boldmath $U(1)$ symmetry}
Having in mind the application to the XXZ chain we would like to include
external fields related to a $U(1)$ symmetry of the model. Let
$\hat \ph \in \End {\mathbb C}^d$ such that $\th (\a) = \re^{\a \hat \ph}$,
$\a \in {\mathbb C}$, satisfies
\begin{equation} \label{uonesym}
     R_{1, 2} (\la, \m) \th_1 (\a) \th_2 (\a) =
        \th_2 (\a) \th_1 (\a) R_{1, 2} (\la, \m) \epp
\end{equation}
Then
\begin{equation}
     [t_\perp (\la), \hat \PH] = 0 \epc
\end{equation}
where $\hat \PH = \sum_{j=1}^L \hat \ph_j$, and we may consider the
Hamiltonian
\begin{equation}
     H = H_0 - h \hat \PH
\end{equation}
with $h \in {\mathbb R}$ instead of $H_0$, as it also commutes with
$t_\perp (\la)$.

\subsection{Local operators}
Sets of consecutive integers will be denoted $\llbracket
j, k \rrbracket$, where $j, k \in {\mathbb Z}$, $j \le k$. Below
we shall consider dynamical correlation functions of two local
operators
\begin{equation} \label{defxy}
     X_{\llbracket 1, \ell\rrbracket} =
        x_1^{(1)} \cdots x_\ell^{(\ell)} \epc \qd
     Y_{\llbracket 1, r\rrbracket} = y_1^{(1)} \cdots y_r^{(r)} \epc
\end{equation}
where $x^{(j)}, y^{(k)} \in \End {\mathbb C}^d$. The numbers $\ell$ and
$r$ will be called the lengths of $X$ and $Y$. We shall assume that these
operators have fixed $U(1)$ charge (or `spin') $s \in {\mathbb C}$,
\begin{equation} \label{spinxy}
     [\hat \PH, X_{\llbracket 1, \ell\rrbracket}] =
        s(X) X_{\llbracket 1, \ell\rrbracket} \epc \qd
     [\hat \PH, Y_{\llbracket 1, r\rrbracket}] =
        s(Y) Y_{\llbracket 1, r\rrbracket} \epp
\end{equation}

For the derivation of the form-factor series we will need
the following representation of the operator $X_{\llb 1, \ell\rrb}$.
\begin{lemma} \label{lem:invinv}
Inversion formula \cite{KMT99a,GoKo00,GKKKS17}.
\begin{equation} \label{sotqip}
     X_{\llbracket 1, \ell\rrbracket} =
        t_\perp^\ell (0) \lim_{\x_j \rightarrow 0}
	\tr_0 \{x^{(\ell)}_0 T_{\perp}^{-1} (\x_\ell)\} \dots
	\tr_0 \{x^{(1)}_0 T_{\perp}^{-1} (\x_1)\} \epp
\end{equation}
\end{lemma}

\subsection{\boldmath The $\s$-staggered monodromy matrix}
\label{sec:sigmastagg}
Fix $M \in {\mathbb N}$. For $j \in \llb 0, M\rrb$ let
$V_j = {\mathbb C}^d$. For $j \in \llb 1, M\rrb$ fix
$\s_j \in \{-1 , 1\}$, $\n_j \in {\mathbb C}$. Let
$\siv = (\s_1, \dots, \s_M)$, $\nuv = (\n_1, \dots, \n_M)$ and
\begin{equation}
     R_{0, j}^{(\s_j)} (\la, \n_j) =
        \begin{cases}
	   R_{0, j} (\la, \n_j) & \text{if $\s_j = 1$} \\[1ex]
	   R_{j, 0}^{t_1} (\n_j, \la) & \text{if $\s_j = - 1$,}
        \end{cases}
\end{equation}
where $t_1$ denotes the transposition with respect to the first
space $R$ is acting on. By definition the $\siv$-staggered monodromy
matrix $T (\la|\siv, \nuv, h) \in \End \bigl(\bigotimes_{j=0}^M V_j\bigr)$
is
\begin{multline} \label{stagmono}
     T (\la|\siv, \nuv, h) =
        \th_0 (h/T) R_{0, M}^{(\s_M)} (\la, \n_M) \dots
	   R_{0, 1}^{(\s_1)} (\la, \n_1) \\[-1ex] = \th_0 (h/T)
        \prod_{j \in \llbracket 1, M\rrbracket}^\dst{\curvearrowleft}
        R_{0, j}^{(\s_j)} (\la, \n_j) \epc
\end{multline}
The arrow above the product indicates descending order.
\begin{lemma} \label{lem:staginv}
Inversion formula in the $\siv$-staggered case.
Let $t (\la|\siv, \nuv, h) = \tr_0 \{T (\la|\siv, \nuv, h)\}$ be the transfer
matrix that is associated with $T (\la|\siv, \nuv, h)$. Then, for any
$x \in \End \bigl({\mathbb C}^d \bigr)$ and any $j \in \llbracket
1, M \rrbracket$ such that $\s_j = 1$,
\begin{equation} \label{inversestagg}
     x_j = \biggl[ \prod_{k=1}^{j-1} t^{\s_k} (\n_k |\siv, \nuv, h) \biggr]
           \tr_0 \{x_0\, T(\n_j|\siv, \nuv, h)\}
           \biggl[ \prod_{k=1}^j t^{- \s_k} (\n_k |\siv, \nuv, h) \biggr] \epc
\end{equation}
provided that $t (\n_k|\siv, \nuv, h)$ is invertible for $k = 1, \dots, j$.
\end{lemma}

A proof of this lemma is provided in appendix~\ref{app:proofinvlemma}
and is a extension of the proof for single-site operators in \cite{GKKKS17}.
For the realization of the dynamical correlation functions at time
$t$ we choose $M = 2N + 2\ell$ in the definition of the $\siv$-staggered
monodromy matrix, where $N$ is the Trotter number and $\ell$ the length
of the operator $X$, and we consider a special staggering $\s^{(\ell)}$
and a special set of inhomogeneities $\nuv^{(\ell)}$,
\begin{subequations}
\label{inhoms}
\begin{align} \label{sigmas}
     \s_{2k - 1}^{(\ell)} & = - \s_{2k}^{(\ell)} = 1 \qqd
     k = 1, \dots, N/2, N/2 + \ell + 1, \dots, N + \ell \epc \notag \\
     \s_{N + k}^{(\ell)} & =
        \begin{cases}
	   + 1 & \text{$k = 1, \dots, \ell$,} \\[.5ex]
	   - 1 & \text{$k = \ell + 1, \dots, 2 \ell$,}
        \end{cases} \\[1ex] \label{nus}
     \n_{2k - 1}^{(\ell)} & = - \n_{2k}^{(\ell)} =
        \begin{cases}
	   \frac{\i t}{\k N} & \text{$k = 1, \dots, N/2$,} \\[1ex]
	   \bigl(\frac{1}{T} - \i t\bigr) 
	   \frac{1}{\k N} & \text{$k = N/2 + \ell + 1, \dots, N + \ell$,}
        \end{cases} \notag \\
     \n_{N + k}^{(\ell)} & =
        \begin{cases}
	   \x_k & \text{$k = 1, \dots, \ell$,} \\[.5ex]
	   0 & \text{$k = \ell + 1, \dots, 2 \ell$.}
        \end{cases}
\end{align}
\end{subequations}
When $\ell = 0$ we simply drop the second lines in (\ref{sigmas})
and (\ref{nus}). For an illustration see Fig.~\ref{fig:cfpathint}.

We shall write $T_{(\ell)} (\la| h) = T(\la|\siv^{(\ell)}, \nuv^{(\ell)}, h)$
and $t_{(\ell)} (\la| h) = t(\la|\siv^{(\ell)}, \nuv^{(\ell)}, h)$.
For the $n$th eigenvalue and eigenvector of the dynamical
quantum transfer matrix $t_{(\ell)} (\la| h)$ we are employing
the notation $\La_{(\ell), n} (\la|h)$, $|\Ps_{(\ell), n} (h)\>$.
We write $T (\la| h) = T_{(0)} (\la| h)$, $t (\la| h) = t_{(0)} (\la| h)$,
$\La_n (\la|h) = \La_{(0), n} (\la|h)$ and $|\Ps_n (h)\>
= |\Ps_{(0), n} (h)\>$ for short. The label $n = 0$ will be reserved
for the dominant state.

\subsection{Two-point functions of local operators}
We define the dynamical two point functions of two
local operators $X$ and $Y$ of length $\ell$ and $r$ as
\begin{multline} \label{defcfxy}
     \bigl\< X_{\llbracket 1, \ell \rrbracket} (t)
        Y_{\llbracket 1 + m, r + m \rrbracket} \bigr\>_T = \\
     \lim_{L \rightarrow \infty} Z_L^{-1} \tr_{1, \dots, L}
	\bigl\{\re^{(\i t - 1/T) H} X_{\llbracket 1, \ell \rrbracket}
	       \re^{- \i t H} Y_{\llbracket 1 + m, r + m \rrbracket} \bigr\} \epc
\end{multline}
where
\begin{equation}
     Z_L = \tr_{1, \dots, L} \bigl\{\re^{- H/T}\bigr\}
\end{equation}
is the partition function. Since the adjoint action of $H$
commutes with the adjoint action of $\hat \PH$, the correlation
function can only be non-zero if
\begin{equation}
     s(X) + s(Y) = 0 \epp
\end{equation}

We shall now use the above defined dynamical quantum transfer
matrix in an argument similar to the one employed in the
derivation of our main theorem in \cite{GKKKS17}. Using (\ref{rhoop}),
(\ref{spinxy}) and (\ref{sotqip}) in (\ref{defcfxy}) we see that
\begin{multline} \label{cfxynext}
     \bigl\< X_{\llbracket 1, \ell \rrbracket} (t)
        Y_{\llbracket 1 + m, r + m \rrbracket} \bigr\>_T = \\
     \re^{- \i t h s(X)} \lim_{L \rightarrow \infty} Z_L^{-1}
        \lim_{N \rightarrow \infty} \lim_{\x_j \rightarrow 0}
	\tr_{1, \dots, L}
	\bigl\{\re^{h \hat \PH/T} \r_{N, L} (1/T - \i t) \,
	       t_\perp^\ell (0) \\ \Bigl[
              \prod_{k \in \llbracket 1, \ell\rrbracket}^\dst{\curvearrowleft}
	      \tr_0 \{x^{(k)}_0 T_{\perp}^{-1} (\x_k)\} \Bigr]
	       \r_{N, L} (\i t) Y_{\llbracket 1 + m, r + m \rrbracket} \bigr\} \epp
\end{multline}
The term under the trace on the right hand side of this equation
is represented graphically in Fig.~\ref{fig:cfpathint}. Our conventions
(\ref{drawr}) for the graphical representation of $R$-matrices are the
same as in \cite{Goehmann20}. The column-to-column transfer matrices in the
\begin{figure}
\begin{center}
\includegraphics[width=.92\textwidth]{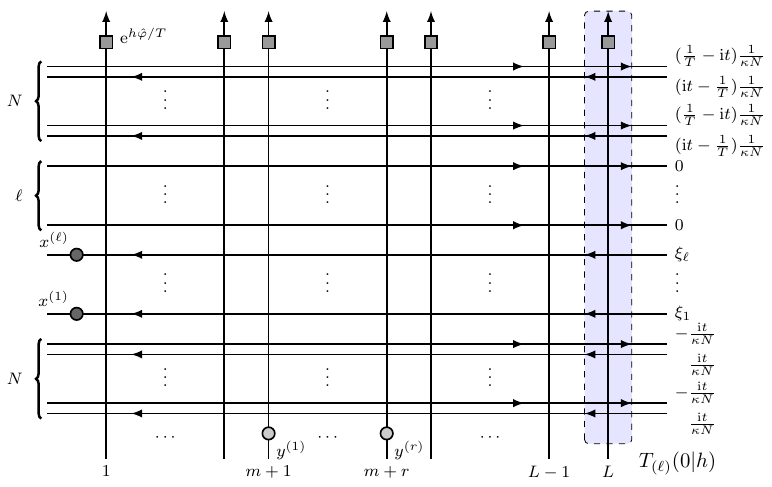}
\caption{\label{fig:cfpathint} Graphical representation
of the expression under the trace in (\ref{cfxynext}). Periodic
boundary conditions at the left and the right are understood. Crosses
represent $R$-matrices, circles and squares the indicated local
operators. Connection of lines means summation over the corresponding
matrix indices. The shaded box on the right highlights the monodromy
matrix $T_{(\ell)}(0|h)$.}
\end{center}
\end{figure}
picture are the dynamical quantum transfer matrices. The corresponding
monodromy matrices are obviously equal to the monodromy matrices
$T_{(\ell)} (\la|h)$ defined in Sec.~\ref{sec:sigmastagg}. We can
therefore read off from the picture that
\begin{multline} \label{cfxynn}
     \bigl\< X_{\llbracket 1, \ell \rrbracket} (t)
        Y_{\llbracket 1 + m, r + m \rrbracket} \bigr\>_T =
     \re^{- \i t h s(X)} \lim_{L \rightarrow \infty}
        \lim_{N \rightarrow \infty} \lim_{\x_j \rightarrow 0} \\[1ex]
	\frac{\tr_{1, \dots, 2N + 2\ell} \Bigl\{
	      X_{\llb N+1, N+\ell \rrb} \, t_{(\ell)}^m (0|h)
	      \Bigl[\prod_{k \in \llbracket 1, r\rrbracket}^\dst{\curvearrowright}
	      \tr_0 \bigl\{y^{(k)}_0 T_{(\ell)} (0|h)\bigr\} \Bigr]
	      t_{(\ell)}^{L - m - r} (0|h)\Bigr\}}
	     {\tr_{1, \dots, 2N + 2\ell} \bigl\{t_{(\ell)}^L (0|h)\bigr\}} \epp
\end{multline}
Making the usual assumptions of the quantum transfer matrix formalism
(that were proven for the XXZ chain at high enough temperature in
\cite{GGKS20}), namely that $|\La_{(\ell),0} (0|h)| > |\La_{(\ell),n} (0|h)|$
for all $n \ne 0$ and that the limits $N \rightarrow \infty$ and
$L \rightarrow \infty$ commute, we conclude that
\begin{multline} \label{cfxynnn}
     \bigl\< X_{\llbracket 1, \ell \rrbracket} (t)
        Y_{\llbracket 1 + m, r + m \rrbracket} \bigr\>_T =
     \re^{- \i t h s(X)} \lim_{N \rightarrow \infty} \lim_{\x_j \rightarrow 0} \\[1ex]
	\frac{\<\Ps_{(\ell), 0} (h)|
	      X_{\llb N+1, N+\ell \rrb} \, t_{(\ell)}^m (0|h)
	      \Bigl[\prod_{k \in \llbracket 1, r\rrbracket}^\dst{\curvearrowright}
	      \tr_0 \bigl\{y^{(k)}_0 T_{(\ell)} (0|h)\bigr\} \Bigr]|\Ps_{(\ell), 0} (h)\>}
	     {\<\Ps_{(\ell),0} (h)|\Ps_{(\ell), 0} (h)\> \La_{(\ell), 0}^{m+r} (0|h)} \epp
\end{multline}
Here we use Lemma~\ref{lem:staginv} and (\ref{inhoms}) to substitute
\begin{multline}
     X_{\llb N + 1, N + \ell\rrb} =
        \Bigl[ \prod_{k=1}^N t_{(\ell)}^{\s_k} (\n_k|h) \Bigr]
	\Bigl[ \prod_{k \in \llbracket 1, \ell \rrbracket}^\dst{\curvearrowright}
	      \tr_0 \bigl\{x^{(k)}_0 T_{(\ell)} (\n_{N+k}|h)\bigr\} \Bigr]
        \Bigl[ \prod_{k=1}^{N + \ell} t_{(\ell)}^{- \s_k} (\n_k|h) \Bigr] \\
	= \Bigl( t_{(\ell)} \bigl(\tst{\frac{\i t}{\k N}}\big|h\bigr) \Bigr)^\frac{N}{2}
	  \Bigl( t_{(\ell)} \bigl(\tst{-\frac{\i t}{\k N}}\big|h\bigr) \Bigr)^{- \frac{N}{2}}
	\Bigl[ \prod_{k \in \llbracket 1, \ell \rrbracket}^\dst{\curvearrowright}
	        \tr_0 \bigl\{x^{(k)}_0 T_{(\ell)} (\x_k|h)\bigr\} \Bigr] \\ \times
	  \Bigl[ \prod_{k=1}^\ell t_{(\ell)}^{-1} (\x_k|h) \Bigr]
	  \Bigl( t_{(\ell)} \bigl(\tst{\frac{\i t}{\k N}}\big|h\bigr) \Bigr)^{- \frac{N}{2}}
	  \Bigl( t_{(\ell)} \bigl(\tst{-\frac{\i t}{\k N}}\big|h\bigr) \Bigr)^{\frac{N}{2}}
	  \epp
\end{multline}
Finally, expanding in terms of the eigenstates of the quantum
transfer matrix, we have arrived at the following Lemma.

\begin{lemma} \label{lem:tfs}
In any fundamental integrable lattice model with $U(1)$ symmetric
$R$-matrix the dynamical two-point functions of two local operators
$X$ and $Y$ of the form (\ref{defxy}) and with $U(1)$ charges
$s(X)$, $s(Y)$ as defined in (\ref{spinxy}) have the thermal form-%
factor series representation
\begin{multline}
     \bigl\< X_{\llbracket 1, \ell \rrbracket} (t)
        Y_{\llbracket 1 + m, r + m \rrbracket} \bigr\>_T =
	\re^{- \i h t s(X)} \\[.5ex] \times
	\lim_{N \rightarrow \infty} \lim_{\x_j \rightarrow 0}
	\sum_n
	\frac{\<\Ps_{(\ell), 0} (h)|
              \prod_{k \in \llbracket 1, \ell\rrbracket}^\dst{\curvearrowright}
	      \tr_0 \{x^{(k)}_0 T_{(\ell)} (\x_k| h)\}|\Ps_{(\ell), n} (h)\>}
	     {\<\Ps_{(\ell), 0} (h)|\Ps_{(\ell), 0} (h)\>
	      \prod_{k=1}^\ell \La_{(\ell), n} (\x_k|h)} \\[1ex] \times
	\frac{\<\Ps_{(\ell), n} (h)|
              \prod_{k \in \llbracket 1, r\rrbracket}^\dst{\curvearrowright}
	      \tr_0 \{y^{(k)}_0 T_{(\ell)} (0| h)\}|\Ps_{(\ell), 0} (h)\>}
	     {\<\Ps_{(\ell), n} (h)|\Ps_{(\ell), n} (h)\> \La_{(\ell), 0}^r (0|h)}
        \biggl(\frac{\La_{(\ell), n} (0|h)}{\La_{(\ell), 0} (0|h)}\biggr)^m
	   \\[1ex] \times
        \biggl(\frac{\La_{(\ell), 0} \bigl(\frac{\i t}{\k N}\big|h\bigr)}
	            {\La_{(\ell), n} \bigl(\frac{\i t}{\k N}\big|h\bigr)}
		       \biggr)^\frac{N}{2}
        \biggl(\frac{\La_{(\ell), n} \bigl(- \frac{\i t}{\k N}\big|h\bigr)}
	            {\La_{(\ell), 0} \bigl(- \frac{\i t}{\k N}\big|h\bigr)}
		       \biggr)^\frac{N}{2} \epp
\end{multline}
\end{lemma}
Note that the order of the products here is ascending.
Lemma~\ref{lem:tfs} is a most natural generalization of our theorem
in \cite{GKKKS17}.

At this point the homogeneous limit can be taken, which brings about
the following simplification.
\begin{theorem} The thermal form-factor expansion.
In any fundamental integrable lattice model with $U(1)$ symmetric
$R$-matrix the dynamical two-point functions of two local operators
$X$ and $Y$ of the form (\ref{defxy}) and with $U(1)$ charges
$s(X)$, $s(Y)$ as defined in (\ref{spinxy}) have the thermal form-%
factor series representation
\begin{multline} \label{tffsxy}
     \bigl\< X_{\llbracket 1, \ell \rrbracket} (t)
        Y_{\llbracket 1 + m, r + m \rrbracket} \bigr\>_T =
	\re^{- \i h t s(X)} \\[.5ex] \times
	\lim_{N \rightarrow \infty}
	\sum_n
	\frac{\<\Ps_0 (h)|
              \prod_{k \in \llbracket 1, \ell\rrbracket}^\dst{\curvearrowright}
	      \tr_0 \{x^{(k)}_0 T (0| h)\}|\Ps_n (h)\>}
	     {\<\Ps_0 (h)|\Ps_0 (h)\> \La_n^\ell (0|h)} \\[1ex] \times
	\frac{\<\Ps_n (h)|
              \prod_{k \in \llbracket 1, r\rrbracket}^\dst{\curvearrowright}
	      \tr_0 \{y^{(k)}_0 T (0| h)\}|\Ps_0 (h)\>}
	     {\<\Ps_n (h)|\Ps_n (h)\> \La_0^r (0|h)}
        \biggl(\frac{\La_n (0|h)}{\La_0 (0|h)}\biggr)^m
	   \\[1ex] \times
        \biggl(\frac{\La_0 \bigl(\frac{\i t}{\k N}\big|h\bigr)}
	            {\La_n \bigl(\frac{\i t}{\k N}\big|h\bigr)}
		       \biggr)^\frac{N}{2}
        \biggl(\frac{\La_n \bigl(- \frac{\i t}{\k N}\big|h\bigr)}
	            {\La_0 \bigl(- \frac{\i t}{\k N}\big|h\bigr)}
		       \biggr)^\frac{N}{2} \epp
\end{multline}
\end{theorem}
The interpretation of this theorem is as follows. The left hand
side can be represented graphically, e.g.\ as in Fig.~\ref{fig:cfpathint}.
The graphical representation may be interpreted as a lattice
path integral with the real or imaginary time direction going
from bottom to top and the space direction going from left to
right. The right hand side of (\ref{tffsxy}) expresses the
path integral entirely in term of `integrable data', namely
in terms of the eigenvalues of the quantum transfer matrix and
certain matrix elements of products of combinations of monodromy
matrix elements.

An interesting property of this representation is that the
features of the environment (or heat bath) are encoded in the
choice of the inhomogeneities (\ref{inhoms}). The form of
the right hand side of (\ref{tffsxy}) remains the same for
different choices of inhomogeneities. This means that the
formula is valid for many different kinds of environments,
e.g.\ for generalized Gibbs ensembles \cite{SaKl03,FaEs13}
or for the ground state of quantum chains of finite length
or for correlation functions in excited states (cf.\ the
discussion in \cite{GKW21}).

Although equation (\ref{tffsxy}) looks quite appealing, in
particular, since it involves only a single sum over excited
states, its actual evaluation remains a formidable task, because
it requires the characterization of all eigenstates of the
quantum transfer matrix (in a given `pseudo-spin sector', see
below) and because efficient formula for matrix elements of
the form appearing in (\ref{tffsxy}) are generally not available,
not even for Yang-Baxter integrable models. This is the reason
why we concentrate in the remainder of this work on the XXZ
quantum spin chain for which we can rely on previous results 
on the spectrum of the quantum transfer matrix \cite{DGKS15b}
and on the evaluation of matrix elements \cite{DGK13a,DGK14a,%
DGKS16b,BGKS21a,BGKSS21}. The essential tool for such an
evaluation is Slavnov's determinant formula \cite{Slavnov89,BeSl19}
for the scalar product of an off-shell and an on-shell
Bethe vector.

\subsection{Two-point functions of elementary blocks}
A basis for the two-point correlation functions is obtained
by considering so-called elementary blocks.
\begin{corollary}
Setting
\begin{equation}
     X_{\llbracket 1, \ell\rrbracket} =
        {e_1}^{\a_1}_{\be_1} \dots {e_\ell}^{\a_\ell}_{\be_\ell} \epc \qd
     Y_{\llbracket 1, r\rrbracket} =
        {e_1}^{\g_1}_{\de_1} \dots {e_r}^{\g_r}_{\de_r}
\end{equation}
in (\ref{tffsxy}) we obtain an expression for the two-point functions
of elementary blocks,
\begin{multline} \label{tffseb}
     \bigl\< \bigl({e_1}^{\a_1}_{\be_1} \dots {e_\ell}^{\a_\ell}_{\be_\ell}\bigr) (t) \,
        {e_{1 + m}}^{\g_1}_{\de_1} \dots {e_{r + m}}^{\g_r}_{\de_r} \bigr\>_T =
	\re^{- \i h t s(X)} \\[.5ex] \times
	\lim_{N \rightarrow \infty}
	\sum_n
	\frac{\<\Ps_0 (h)| T^{\a_1}_{\be_1} (0| h)
	                   \dots T^{\a_\ell}_{\be_\ell} (0|h) |\Ps_n (h)\>}
	     {\<\Ps_0 (h)|\Ps_0 (h)\> \La_n^\ell (0|h)} \\[1ex] \times
	\frac{\<\Ps_n (h)| T^{\g_1}_{\de_1} (0| h)
	                   \dots T^{\g_r}_{\de_r} (0|h) |\Ps_0 (h)\>}
	     {\<\Ps_n (h)|\Ps_n (h)\> \La_0^r (0|h)}
        \biggl(\frac{\La_n (0|h)}{\La_0 (0|h)}\biggr)^m
	   \\[1ex] \times
        \biggl(\frac{\La_0 \bigl(\frac{\i t}{\k N}\big|h\bigr)}
	            {\La_n \bigl(\frac{\i t}{\k N}\big|h\bigr)}
		       \biggr)^\frac{N}{2}
        \biggl(\frac{\La_n \bigl(- \frac{\i t}{\k N}\big|h\bigr)}
	            {\La_0 \bigl(- \frac{\i t}{\k N}\big|h\bigr)}
		       \biggr)^\frac{N}{2} \epc
\end{multline}
where $T^\a_\be (\la|h) = \tr_0 \bigl\{{e_0}^\a_\be\, T(\la|h)\bigr\}$.
\end{corollary}
Mathematically this seems to be a very natural object. We will discuss
it in great detail below, when we consider the concrete example of the
XXZ chain.
%Note that, strictly speaking,
%$s(X)$ may be undefined here. The well-definiteness has to be
%stipulated as a property of the local gauge field $\hat \ph$.

\section{Two-point functions of spin-zero operators in the
XXZ chain}
\label{sec:spin_zero_xxz}
We shall now fill the general formulae obtained in the
previous section with life by considering the concrete
example of the XXZ spin chain. In the course of this attempt
we will restrict ourselves to spin-zero operators, then
identify properly normalized form factors in the summands
in (\ref{tffseb}) and connect the latter with constructions
that have been introduced in the context of the reduced
density matrix of the XXZ chain. These are discrete functional
equations of rqKZ type, multiple-integral representations
and the Fermionic basis. We shall see that the properly
normalized form factors factorize and can be expressed
by only two functions related to the form factors of
local operators of length one and two.
\subsection{Hamiltonian and integrable structure}
The XXZ or Heisenberg-Ising chain is described by the
Hamiltonian
\begin{equation} \label{hxxz}
     H_0 = J \sum_{j = 1}^L \Bigl\{ \s_{j-1}^x \s_j^x + \s_{j-1}^y \s_j^y
                 + \D \bigl( \s_{j-1}^z \s_j^z - 1 \bigr) \Bigr\} \epc
		 %- \frac{h}{2} \sum_{j=1}^L \s_j^z \epc
\end{equation}
where the operators $\s^\a \in \End {\mathbb C}^2$, $\a = x, y, z$,
are represented by the Pauli matrices. The two real parameters in
(\ref{hxxz}) are the exchange (or Heisenberg) interaction $J > 0$
and the anisotropy $\D$ (or the Ising interaction $J \D$). In the
case of the XXZ chain we may choose
\begin{equation}
     \hat \PH = S^z = \2 \sum_{j=1}^L \s_j^z
\end{equation}
or $\hat \ph = \2 \s^z$.

The integrable structure of the model is encoded in its $R$-matrix.
Having in mind applications to the massive antiferromagnetic regime
we parameterise it as
\begin{equation} \label{rmatrix}
     \begin{array}{cc}
     R(\la,\m) = \begin{pmatrix}
                  1 & 0 & 0 & 0 \\
		  0 & b(\la,\m) & c(\la,\m) & 0 \\
		  0 & c(\la,\m) & b(\la,\m) & 0 \\
		  0 & 0 & 0 & 1
		 \end{pmatrix} \epc &
     \begin{array}{c}
     b(\la, \m) = \frac{\sin(\m - \la)}{\sin(\m - \la + \i \g)} \\[2ex]
     c(\la, \m) = \frac{\sin(\i \g)}{\sin(\m - \la + \i \g)}
    \end{array}
    \end{array} \epp
\end{equation}
We set $q = \re^{- \g}$. Then $\D = (q + q^{-1})/2$ and
$\k = - \i/[J (q - q^{-1})]$ (which follows from (\ref{fundham})
and (\ref{hxxz})).

The parameterization in (\ref{rmatrix}) ensures that both, equations
(\ref{ybe}) and (\ref{fund}), are satisfied. In addition, the
$R$-matrix (\ref{rmatrix}) satisfies the crossing relation
\begin{equation} \label{crisscross}
     R_{2,1}^{t_1} (\m, \la) =
        \frac{\s_1^y R_{1, 2} (\la + \i \g, \m) \s_1^y}{b(\la + \i \g, \m)} =
        \frac{\s_2^y R_{1, 2} (\la, \m - \i \g) \s_2^y}{b(\la, \m - \i \g)}
\end{equation}
that will be needed below when we discuss the properties
of the spin-zero thermal form factors of the XXZ chain.

\subsection{Spin-zero operators and states}
As an example for the formalism of the previous section we shall
consider two-point functions of arbitrary spin-zero operators $X$
and $Y$ in the XXZ chain. This will give us new insight into the
structure of the general two-point functions and will connect our
work with the theory of factorizing correlation functions in the
static case \cite{BJMST08a,BoGo09,JMS08}.

Examples of relevant spin-zero operators are
\begin{subequations}
\begin{align}
     & \tst{\2} \s^z && \text{`magnetization'} \epc \\[1ex]
     & \i (\s^- \otimes \s^+ - \s^+ \otimes \s^-) && \text{`magnetic current'} \epc \\[1ex]
     & 2 (\s^- \otimes \s^+ + \s^+ \otimes \s^-) + \D \s^z \otimes \s^z
       && \text{`energy'} \epp \label{energy}
\end{align}
\end{subequations}
The two-point functions of local magnetic currents, in particular,
determine the so-called optical conductivity and the magnetic
Drude weight of the model, which have been the subject of much
debate in the past.

The conservation of the $z$-component of the total spin by the
XXZ Hamiltonian translates into a `pseudo-spin $z$-component'
conservation of the corresponding quantum transfer matrix. If
we define the pseudo-spin operator
\begin{equation} \label{pseudospinop}
     \Si = \2 \sum_{j=1}^N (-1)^{j-1} \s_j^z \epc
\end{equation}
then
\begin{equation} \label{psconservation}
     [t(\la|h), \Si] = 0 \epp
\end{equation}

Denoting the pseudo spin by $s_p$ we have the relation
\begin{equation} \label{sequalsps}
     s_p \bigl(T^\a_\be (\la|h)\bigr) = s(e^\a_\be) = \frac{\be - \a}{2} \epc
\end{equation}
for $\a, \be \in \{-1, 1\}$. A brief derivation is provided
in Appendix~\ref{app:pseudospin}. The dominant state has pseudo
spin $s_p = 0$. If $|\Ps_n (h')\>$ is an excited state with
pseudo spin $s_p'$, then it holds that $\<\Ps_0 (h)|\Ps_n (h')\> = 0$,
unless $s_p' = 0$. In the latter case the pairing vanishes at $h = h'$,
but is generically non-zero. Restricting ourselves to operators
of spin zero therefore means to impose the restrictions
\begin{equation}
     \sum_{j=1}^\ell (\be_j - \a_j) = 0 = \sum_{j=1}^r (\g_j - \de_j)
\end{equation}
onto the elementary blocks in (\ref{tffseb}).

\subsection{Properly normalized form factors}
In the following section we will utilize our results on the
form factors of the longitudinal correlation functions of
the XXZ chain in the antiferromagnetic massive regime \cite{BGKS21a,%
BGKSS21}. For this reason we have adopted our conventions for
the $R$-matrix, the Hamiltonian and the monodromy matrix from
\cite{BGKS21a}. In \cite{BGKS21a,BGKSS21} we calculated the
`amplitude'
\begin{equation} \label{ampl}
     A_n (h, h') = \frac{\<\Ps_0 (h)|\Ps_n (h')\>\<\Ps_n (h')|\Ps_0 (h)\>}
                        {\<\Ps_0 (h)|\Ps_0 (h)\>\<\Ps_n (h')|\Ps_n (h')\>}
\end{equation}
and its second derivative with respect to
\begin{equation} \label{defalpha}
     \a = \frac{h - h'}{2 \g T}
\end{equation}
at $\a = 0$ explicitly in the low-$T$ limit for $\D > 1$ and
magnetic fields below the lower critical field. Combing this
with the results obtained for the eigenvalue ratios
\begin{equation} \label{defrho}
     \r_n (\la|h, h') = \frac{\La_n (\la|h')}{\La_0 (\la|h)}
\end{equation}
in the same limit, we obtained an explicit series representation
for the longitudinal correlation function (the two-point function
of the local magnetization).

In the light of our work \cite{BGKS21a,BGKSS21} the amplitude
$A_n (h,h')$ seems a rather natural object, worth to be kept
in the calculation of more general correlation functions as well.
This suggests to define the matrix valued functions on $\ell$-fold,
resp.\ $r$-fold, tensor products of auxiliary spaces
\begin{subequations}
\label{propforms}
\begin{align}
     {\cal F}_{n; \ell}^{(-)} (\x_1, \dots,\x_\ell|h, h') & =
        \frac{\<\Ps_0 (h)|T(\x_1|h') \otimes \dots \otimes T(\x_\ell|h')|\Ps_n (h')\>}
             {\<\Ps_0 (h)|\Ps_n (h')\> \prod_{j=1}^\ell \La_n (\x_j|h')} \epc \\[1ex]
     {\cal F}_{n; r}^{(+)} (\z_1, \dots,\z_r|h, h') & =
        \frac{\<\Ps_n (h')|T(\z_1|h) \otimes \dots \otimes T(\z_r|h)|\Ps_0 (h)\>}
             {\<\Ps_n (h')|\Ps_0 (h)\> \prod_{j=1}^r \La_0 (\z_j|h)} \epp
\end{align}
\end{subequations}

\begin{corollary}
Using these functions the two-point functions of spin-zero
elementary blocks (\ref{tffseb}) can be written as
\begin{multline} \label{tffszero}
     \bigl\< \bigl({e_1}^{\a_1}_{\be_1} \dots {e_\ell}^{\a_\ell}_{\be_\ell}\bigr) (t) \,
        {e_{1 + m}}^{\g_1}_{\de_1} \dots {e_{r + m}}^{\g_r}_{\de_r} \bigr\>_T = \\
	\lim_{N \rightarrow \infty} \;
	\lim_{h' \rightarrow h} \;
	\lim_{\x_j, \z_k \rightarrow 0} \;
     \sum_n A_n (h, h') \r_n (0|h, h')^m
     \biggl(
     \frac{\r_n \bigl(- \frac{\i t}{\k N}\big|h, h')}
          {\r_n \bigl(\frac{\i t}{\k N}\big|h, h')}
	  \biggr)^\frac{N}{2} \\[1ex] \times
     {{\cal F}_{n; \ell}^{(-)}}^{\a_1 \dots \a_\ell}_{\be_1 \dots \be_\ell}
        (\x_1, \dots,\x_\ell|h, h') \,
     {{\cal F}_{n; r}^{(+)}}^{\g_1 \dots \g_r}_{\de_1 \dots \de_r}
        (\z_1, \dots,\z_r|h, h') \epp
\end{multline}
\end{corollary}
As will be discussed in the remaining part of this work, this
structure seems to be utterly useful. We will argue, in particular,
that the functions ${\cal F}_{n; m}^{(\pm)}$ defined in (\ref{propforms})
are properly normalized thermal form factors on the lattice. In the
following we refer to them as `the thermal form factors'.

For $n = 0$ the thermal form factors reduce to the generalized
reduced density matrix
\begin{equation} \label{defreddens}
     {\cal D}_m (\x_1, \dots,\x_m|h, h') =
        {\cal F}_{0; m}^{(-)} (\x_1, \dots,\x_m|h, h') =
        {\cal F}_{0; m}^{(+)} (\x_1, \dots,\x_m|h', h)
\end{equation}
studied intensively in the literature by means of the
algebraic Bethe ansatz \cite{BoGo09} and by `the Fermionic
basis approach' \cite{BJMST04b,BJMST06b,BJMST08a,BJMS09a,JMS08},
a method based on the study of functional equations and
on the representation theory of quantum groups (see also
\cite{GKW21,Goehmann20} for pedagogical partial accounts). 
The spirit behind the Fermionic basis approach is to
reconstruct the reduced density matrix from its properties.
The same is possible for the closely related thermal form
factors above. As in case of the density matrix it is
useful to stay more general and return to the $\siv$-staggered
inhomogeneous monodromy matrix (\ref{stagmono}). The
eigenvectors and eigenvalues of the corresponding transfer
matrix $t(\la|\siv, \nuv, h)$ will be denoted $|\Ps_n (\siv, \nuv, h)\>$
and $\La_n (\la|\siv, \nuv, h)$. We shall assume that $M$ is
even (see (\ref{stagmono})) and we restrict ourselves to
states with pseudo-spin zero. For every $\xiv = (\x_1, \dots,
\x_m) \in {\mathbb C}^m$ the functions 
\begin{subequations}
\label{geninhomff}
\begin{align}
     {\cal F}_{n; m}^{(-)} (\xiv|\siv, \nuv, h, h') & =
        \frac{\<\Ps_0 (\siv, \nuv, h)|T(\x_1|\siv, \nuv, h')
	      \otimes \dots \otimes T(\x_m|\siv, \nuv, h')|\Ps_n (\siv, \nuv, h')\>}
             {\<\Ps_0 (\siv, \nuv, h)|\Ps_n (\siv, \nuv, h')\>
	      \prod_{j=1}^m \La_n (\x_j|\siv, \nuv, h')} \epc \\[1ex]
     {\cal F}_{n; m}^{(+)} (\xiv|\siv, \nuv, h, h') & =
        \frac{\<\Ps_n (\siv, \nuv, h')|T(\x_1|\siv, \nuv, h)
	      \otimes \dots \otimes T(\x_m|\siv, \nuv, h)|\Ps_0 (\siv, \nuv, h)\>}
             {\<\Ps_n (\siv, \nuv, h')|\Ps_0 (\siv, \nuv, h)\>
	      \prod_{j=1}^m \La_0 (\x_j|\siv, \nuv, h)}
\end{align}
\end{subequations}
are then well-defined generalizations of the thermal form
factors (\ref{propforms}). They reduce to (\ref{propforms})
if we substitute $\siv = \siv^{(0)}, \nuv = \nuv^{(0)}$ (see
(\ref{inhoms})). We also extend the definition (\ref{defrho})
to the general $\siv$-staggered and inhomogeneous case,
\begin{equation} \label{defrhogen}
     \r_n (\la|\siv, \nuv, h, h') =
        \frac{\La_n (\la|\siv, \nuv, h')}{\La_0 (\la|\siv, \nuv, h)} \epp
\end{equation}

\subsection{Properties of the thermal form factors}
The above defined generalized thermal form factors connect
our work on dynamical correlation functions with previous
work on factorized correlation functions of the XXZ chain
\cite{BGKS07,BoGo09,BJMST08a,JMS08} which was based on
the properties of a generalized reduced density matrix
(equation (\ref{defreddens}) in the $\siv$-staggered inhomogeneous
case). The properties of both objects are very similar. 
In order to state these properties let us fix some more
notation.

We first of all recall the natural right action of the
symmetric group on row vectors,
\begin{equation}
     \nuv P = (\n_{P1}, \dots, \n_{PM})
\end{equation}
for all $P \in \mathfrak{S}^M$. The neighbour transposition,
a special permutation that interchanges $j$ and $j+1$,
will be denoted $\Pi_{j, j+1}$. We further introduce the maps
$\iota_j$ and $S_j$,
\begin{equation}
     (\siv \iota_j)_k = \begin{cases}
                         - \s_k & \text{if $k = j$} \\
			   \s_k & \text{if $k \ne j$,}
                      \end{cases} \qqd
     (\nuv S_j)_k = \begin{cases}
                     \n_k + \i \s_k \g & \text{if $k = j$} \\
		     \n_k & \text{if $k \ne j$,}
                  \end{cases}
\end{equation}
and the special staggering $\siv_-$ defined by $(\s_-)_k = - 1$
for all $k \in \llb 1, M \rrb$. In addition, we shall need
the inhomogeneous and twisted version of the row-to-row monodromy
matrix (\ref{monoperp}),
\begin{equation}
     T_{\perp} (\la|\xiv, h) =
        R_{0, m} (\la, \x_m) \dots R_{0, 1} (\la, \x_1) \th_0 (h/T) \epp
\end{equation}
\begin{lemma} \label{lem:propffs}
Properties of the spin-zero generalized thermal form factors.
\begin{enumerate}
\item
Normalization condition.
\begin{equation}
     \tr_{1, \dots, m}
        \bigl\{{\cal F}_{n; m}^{(\pm)} (\xiv|\siv, \nuv, h, h')\bigr\} = 1 \epp
\end{equation}
\item
Reduction relations.
\begin{subequations}
\label{redurel}
\begin{align}
     & \tr_m \bigl\{{\cal F}_{n; m}^{(\pm)} (\xiv|\siv, \nuv, h, h')\bigr\} =
              {\cal F}_{n; m - 1}^{(\pm)} ((\x_1, \dots, \x_{m - 1})|\siv, \nuv, h, h')
	      \epc \\[1ex]
     & \tr_1 \bigl\{q^{\pm \a \s_1^z}
           {\cal F}_{n; m}^{(\pm)} (\xiv|\siv, \nuv, h, h')\bigr\} \notag \\ 
	   & \mspace{126.mu} = \r_n^{\pm 1} (\x_1|\siv, \nuv, h, h') \, 
	   {\cal F}_{n; m - 1}^{(\pm)} ((\x_2, \dots, \x_m)|\siv, \nuv, h, h')
\end{align}
\end{subequations}
with $\a$ as in (\ref{defalpha}).
\item
Exchange relation. Let $\check R = PR$. Then
\begin{multline} \label{exchangerel}
     \check R_{j, j+1} (\x_j, \x_{j+1})
        {\cal F}_{n; m}^{(\pm)} (\xiv|\siv, \nuv, h, h') \\ =
        {\cal F}_{n; m}^{(\pm)} (\xiv \Pi_{j, j+1}|\siv, \nuv, h, h')
        \check R_{j, j+1} (\x_j, \x_{j+1})
\end{multline}
for $j \in \llb 1, m - 1\rrb$.
\item
$U(1)$ symmetry. For any $\k \in {\mathbb C}$
\begin{equation}
     \bigl[{\cal F}_{n; m}^{(\pm)} (\xiv|\siv, \nuv, h, h'),
        \bigl(\th(\k) \bigr)^{\otimes m}\bigr] = 0 \epp
\end{equation}
\item
Row reflection (`crossing').
\begin{equation}
     {\cal F}_{n; m}^{(\pm)} (\xiv|\siv, \nuv, h, h')
        = {\cal F}_{n; m}^{(\pm)} (\xiv|\siv \iota_j, \nuv S_j, h, h')
\end{equation}
for all $j \in \llb 1, M \rrb$.
\item
Commutativity of rows.
\begin{equation}
     {\cal F}_{n; m}^{(\pm)} (\xiv|\siv, \nuv, h, h')
        = {\cal F}_{n; m}^{(\pm)} (\xiv|\siv P, \nuv P, h, h')
\end{equation}
for all $P \in \mathfrak{S}^M$.
\item
Transposition property.
\begin{multline} \label{transprop}
     {{\cal F}_{n; m}^{(-)}}_{\be_1, \dots, \be_m}^{\a_1, \dots, \a_m}
        (\xiv|\siv, \nuv, h, h') =
     \Bigl[\prod_{j=1}^m \r_n^{-1} (\x_j|\siv, \nuv, h, h')\Bigr] \\ \times
     \bigl((q^{\a \s_z})^{\otimes m}
     {\cal F}_{n; m}^{(+)}\bigr)_{\a_m, \dots, \a_1}^{\be_m, \dots, \be_1}
        ((\x_m, \dots, \x_1)|\siv, \nuv, h, h') \epp
\end{multline}
\item
The functions ${\cal F}_{n; m}^{(\pm)} (\xiv|\siv, \nuv, h, h')$ are
meromorphic in all $\x_j$, $j \in \llb 1, m \rrb$.
\item
Asymptotic behaviour.
\begin{subequations}
\begin{align}
     & \lim_{\Im \x_m \rightarrow \pm \infty}
        {\cal F}_{n; m}^{(+)} (\xiv|\siv, \nuv, h, h') \notag \\[-2ex] & \mspace{135.mu}
        = {\cal F}_{n; m - 1}^{(+)} ((\x_1, \dots, \x_{m-1})|\siv, \nuv, h, h')
	  \frac{\th_m \bigl(\frac{h}{T}\bigr)}
	       {\tr \bigl\{\th \bigl(\frac{h}{T}\bigr)\bigr\}} \epc \\
     & \lim_{\Im \x_m \rightarrow \pm \infty}
        {\cal F}_{n; m}^{(-)} (\xiv|\siv, \nuv, h, h') \notag \\[-2ex] & \mspace{135.mu}
        = {\cal F}_{n; m - 1}^{(-)} ((\x_1, \dots, \x_{m-1})|\siv, \nuv, h, h')
	  \frac{\th_m \bigl(\frac{h'}{T}\bigr)}
	       {\tr \bigl\{\th \bigl(\frac{h'}{T}\bigr)\bigr\}} \epp
\end{align}
\end{subequations}
\item
Discrete form of the reduced $q$-Knizhnik-Zamolodchikov equation.
The functions ${\cal F}_{n; m}^{(\pm)}$ satisfy the `discrete
functional equations'
\begin{align} \label{rqKZ}
     & {\cal F}_{n; m}^{(\pm)}
          \bigl((\x_1, \dots, \x_{m-1}, \x_m - \i \g)
	         \big|\siv_-, \nuv, h, h'\bigr) =
        \r_n^{\mp 1} (\x_m|\siv_-, \nuv, h, h') \notag \\[1ex] & \times
	\tr_0 \bigl\{T_{\perp}^{-1} (\x_m|\xiv, h')
	             {\cal F}_{n; m}^{(\pm)} (\xiv|\siv_-, \nuv, h, h')
		     \s_0^y P_{0,m} \s_0^y T_{\perp} (\x_m|\xiv, h) \bigr\} \epc
\end{align}
if $\x_m = \n_1$.
\end{enumerate}
\end{lemma}
\begin{proof}
Properties (i)-(iii) are direct consequences of the definition
(\ref{propforms}).

In order to prove (iv) we rewrite (\ref{pseudospingroup})
in the form
\begin{equation}
     T(\la|h) \th(\k) \re^{\k \Si} = \re^{\k \Si} \th(\k) T(\la|h)
\end{equation}
with $\Si$ according to (\ref{pseudospinsstag}). Then we use that
the states in the definition of ${\cal F}_{n; m}^{(\pm)}$ have spin
zero by hypothesis, whence
\begin{equation}
     \re^{\k \Si} |\Ps_n (h)\> = |\Ps_n (h)\> \epc \qd
     \<\Ps_n (h)| \re^{- \k \Si} = \<\Ps_n (h)| \epp
\end{equation}

(v) is a consequence of the crossing relation (\ref{crisscross}).
This can, for instance, be seen by representing ${\cal F}_{n; m}^{(\pm)}$,
including the states involved in its definition, graphically (cf.\ e.g.\
\cite{Goehmann20} for an introduction to the graphical method).

In a similar way we can understand (vi). We represent the
left hand side graphically and use (v) in order to rewrite
it in such a way that all horizontal lines are oriented to the
left. Then commutativity of the rows follows from the
Yang-Baxter algebra relations with $R$-matrix (\ref{rmatrix}),
which imply that $[T_+^+ (\la|\siv, \nuv, h), T_+^+ (\m|\siv, \nuv, h)]
= 0$. Finally, (v) is used to restore the original $\siv$-staggering.

(vii) can again most easily be shown in the language of the
graphical representation. The proof relies on the symmetry
(\ref{rsym}) of the $R$-matrix which implies that pictorial
expressions stay invariant if the direction of all lines
is reversed simultaneously. In addition, (iv) and (vi) have
to be employed.

(viii) is obvious. (ix) holds because
\begin{equation}
     \lim_{\Im \la \rightarrow \pm \infty} R_{0, j} (\la|\n_j)
        = q^{\mp\2} \bigl({e_0}_+^+ q^{\pm \2 \s_j^z} 
                          + {e_0}_-^- q^{\mp \2 \s_j^z}\bigr)
\end{equation}
and therefore
\begin{equation}
     \lim_{\Im \la \rightarrow \pm \infty} T (\la|\siv, \nuv, h) =
        \th_0 \Bigl(\frac{h}{T}\Bigr) q^{\mp \2 \sum_{j=1}^M \s_j}
	\bigl({e_0}_+^+ q^{\pm \Si} + {e_0}_-^- q^{\mp \Si}\bigr) \epp
\end{equation}

For the proof of (x) one may again use graphical notation.
Then one can closely follow \cite{AuKl12}.
\end{proof}
\begin{remark}
We would like to stress once more that the properties of the generalized
thermal form factors are very similar to the properties of the
generalized reduced density matrix. The information about the excited
states enters the equations in Lemma~\ref{lem:propffs} only through
the function $\r_n$ which appears only in eqs.\ (\ref{redurel}),
(\ref{transprop}) and (\ref{rqKZ}). In particular, the modification
of the rqKZ equation (\ref{rqKZ}) by such a factor appears new and
remarkable to us. Notice that this factor is non-trivial even for
$n = 0$, unless the regularization parameter $\a$ is sent to
zero.
\end{remark}

\subsection{Thermal form factors at work -- the simplest examples}
After the rather general mathematical considerations of the previous
subsection we now come back to the concrete case of thermal form
factors realized by the special $\siv$-staggering (\ref{inhoms}).
It is well known \cite{BGKS07,BoGo09} that the reduction relations
(\ref{redurel}) and the exchange relation (\ref{exchangerel}) are
sufficient to calculate the simplest form factors. In particular,
the form factors of `ultra local operators', $\ell = r = 1$, follow
directly from the reduction relations.

Let us perform this simple but instructive exercise. We set $\ell =
r = 1$ in the reduction relations (\ref{redurel}) and write for short
\begin{equation}
     {F_p}^\a_\be = {{\cal F}_{n; 1}^{(+)}}^\a_\be (\x|h, h') \epc \qd
     {F_m}^\a_\be = {{\cal F}_{n;1}^{(-)}}^\a_\be (\x|h, h') \epp
\end{equation}
Then the reduction relations become
\begin{subequations}
\begin{align}
     {F_p}^+_+ + {F_p}^-_- & = 1 \epc
     & q^\a {F_p}^+_+ + q^{- \a} {F_p}^-_- & = \r_n \epc \\[1ex]
     {F_m}^+_+ + {F_m}^-_- & = 1 \epc
     & q^{- \a} {F_m}^+_+ + q^\a {F_m}^-_- & = 1/\r_n \epc
\end{align}
\end{subequations}
implying that
\begin{equation} \label{fpp1}
     {F_p}^+_+ = 1 - {F_p}^-_- = \frac{\r_n - q^{- \a}}{q^\a - q^{- \a}} \epc \qd
     {F_m}^+_+ = 1 - {F_m}^-_- = \frac{q^\a - 1/\r_n}{q^\a - q^{- \a}} \epp
\end{equation}
Thus, we obtain the following expressions for the thermal form
factors of the magnetization operators,
\begin{subequations}
\label{magpm}
\begin{align}
     \tr\bigl\{\tst{\2} \s^z \, {\cal F}_{n; 1}^{(-)} (\x|h, h')\bigr\} & =
        \frac{\2(q^\a + q^{- \a}) - 1/\r_n (\x|h, h')}{q^\a - q^{- \a}}
	\epc \\[1ex]
     \tr\bigl\{\tst{\2} \s^z \, {\cal F}_{n; 1}^{(+)} (\x|h, h')\bigr\} & =
        \frac{\r_n (\x|h, h') - \2(q^\a + q^{- \a})}{q^\a - q^{- \a}}
	\epp
\end{align}
\end{subequations}

Using the definition (\ref{defalpha}) of $\a$ we conclude that,
for all $n > 0$,
\begin{multline}
     \lim_{h' \rightarrow h}
     \lim_{\x, \z \rightarrow 0} A_n (h, h')
     \tr\bigl\{\s^z \, {\cal F}_{n; 1}^{(-)} (\x|h, h')\bigr\}
     \tr\bigl\{\s^z \, {\cal F}_{n; 1}^{(+)} (\z|h, h')\bigr\} \\[.5ex] =
     2 T^2 \bigl(\6_{h'}^2 A_n (h, h')\bigr)\bigr|_{h' = h}
     \bigl(\r_n (0|h, h) - 2 + 1/\r_n(0|h, h)\bigr) \epc
\end{multline}
which is (up to a shift of the argument of $\r_n$ by $\i \g/2$ that 
occurs because of a change of our conventions) what we had in equation
(30) of our paper \cite{BGKS21a}. Thus, we have successfully reproduced
the formula for the longitudinal correlation function. If $n = 0$ in
(\ref{magpm}), then, for $\x = \z = 0$, both expressions reduce to the
magnetization density in the limit $h' \rightarrow h$, as it should be.

As we know from our previous work \cite{BGKS07,BoGo09} the reduction
relations combined with the exchange relation can be used to calculate
the expectation values of some of the neighbour two-point functions.
We shall see that this can be promoted to the level of thermal
form factors. We set $\ell = r = 2$ in the reduction relations and
use the shorthand notation
\begin{equation}
     {F_p}^{\a \g}_{\be \de} =
        {{\cal F}_{n; 2}^{(+)}}^{\a \g}_{\be \de} (\x_1, \x_2|h, h') \epc \qd
     {F_m}^{\a \g}_{\be \de} =
        {{\cal F}_{n; 2}^{(-)}}^{\a \g}_{\be \de} (\x_1, \x_2|h, h') \epp
\end{equation}
Furthermore, a bar over a function symbol will mean that $\x_1$ and $\x_2$
are interchanged. The reduction relations then imply that
\begin{subequations}
\begin{align}
     q^\a {F_p}^{+-}_{+-} - q^{- \a} {F_p}^{-+}_{-+} & =
        q^\a {F_p}^+_+ - \r_n {\overline{F}_p}^+_+ \epc \\[1ex]
     q^{- \a} {F_m}^{+-}_{+-} - q^\a {F_m}^{-+}_{-+} & =
        q^{- \a} {F_m}^+_+ - \r_n^{-1} {\overline{F}_m}^+_+ \epp
\end{align}
\end{subequations}
Inserting (\ref{fpp1}) we obtain
\begin{subequations}
\label{diagreduced}
\begin{align}
     (q^\a - q^{- \a})
     \bigl(q^\a {F_p}^{+-}_{+-} - q^{- \a} {F_p}^{-+}_{-+}\bigr) & =
        (q^\a + q^{- \a}) \r_n -  1 - \r_n \overline{\r}_n \epc \\[1ex]
     (q^\a - q^{- \a})
     \bigl(q^{- \a} {F_m}^{+-}_{+-} - q^\a {F_m}^{-+}_{-+}\bigr) & =
        1 + \r_n^{-1} \overline{\r}_n^{\, -1} - (q^\a + q^{- \a}) \r_n^{-1} \epp
\end{align}
\end{subequations}
On the other hand, we can infer from the exchange relation that
\begin{equation} \label{diagexchanged}
     2 \bigl(F^{+-}_{-+} - F^{-+}_{+-}\bigr) =
        (1/g + g) \bigl(F^{+-}_{+-} - F^{-+}_{-+}\bigr) +
        (1/g - g) \bigl(\overline F^{+-}_{+-} - \overline F^{-+}_{-+}\bigr) \epc
\end{equation}
where $F = F_p, F_m$ and
\begin{equation}
     g(\x_1, \x_2) = \frac{\sin(\i \g)}{\sin (\x_2 - \x_1)} \epp
\end{equation}

Now, asymptotically for $\a \rightarrow 0$,
\begin{equation}
     F^{\be \g}_{\de \e} \sim \frac{1}{\a} \epc
\end{equation}
where $F = F_p, F_m$, with the simple pole stemming from the pairing
$\<\Ps_0 (h)|\Ps_n (h')\>$ or $\<\Ps_n (h')|\Ps_0 (h)\>$ occurring in
the denominator in the definition of ${\cal F}^{(\pm)}_{n; 2}$. Thus,
combining (\ref{diagreduced}) and (\ref{diagexchanged}) we obtain,
asymptotically for $\a \rightarrow 0$, that
\begin{subequations}
\begin{align}
     {F_p}^{+-}_{-+} - {F_p}^{-+}_{+-} & \sim
        \frac{1}{q^\a - q^{- \a}}
	\bigl[g (\r_n - \overline \r_n)
	       - (1/g) (1 - \r_n)(1 - \overline \r_n)\bigr]_{h' = h} \epc \\[1ex]
     {F_m}^{+-}_{-+} - {F_m}^{-+}_{+-} & \sim
        \frac{1}{q^\a - q^{- \a}} \frac{1}{\r_n \overline \r_n}
	\bigl[g (\r_n - \overline \r_n)
	       + (1/g) (1 - \r_n)(1 - \overline \r_n)\bigr]_{h' = h} \epp
\end{align}
\end{subequations}
This further simplifies by taking the homogeneous limit,
\begin{subequations}
\label{spincurffs}
\begin{align}
     \lim_{\x_2 \rightarrow \x_1}
     \tr \bigl\{\i (\s_1^- \s_2^+ - \s_1^+ \s_2^-)
                {\cal F}_{n; 2}^{(+)} (\x_1, \x_2|h, h')\bigr\} & \sim
        - \frac{\sh(\g) \r_n' (\x_1|h, h)}{q^\a - q^{- \a}} \epc \\[1ex]
     \lim_{\x_2 \rightarrow \x_1}
     \tr \bigl\{\i (\s_1^- \s_2^+ - \s_1^+ \s_2^-)
                {\cal F}_{n; 2}^{(-)} (\x_1, \x_2|h, h')\bigr\} & \sim
        \frac{\sh(\g) \6_{\x_1}1/ \r_n (\x_1|h, h)}{q^\a - q^{- \a}} \epp
\end{align}
\end{subequations}
Finally, assembling all the pieces, we end up with
\begin{multline} \label{curcurfinite}
     \lim_{h' \rightarrow h} \lim_{\x_j, \z_k \rightarrow 0}
      A_n (h, h') \\ \times
     \tr \bigl\{\i (\s_1^- \s_2^+ - \s_1^+ \s_2^-)
                {\cal F}_{n; 2}^{(-)} (\x_1, \x_2|h, h')\bigr\}
     \tr \bigl\{\i (\s_1^- \s_2^+ - \s_1^+ \s_2^-)
                {\cal F}_{n; 2}^{(+)} (\z_1, \z_2|h, h')\bigr\} \\[1ex] =
     \frac{\sh^2 (\g)  T^2}{2} \bigl(\6_{h'}^2 A_n (h, h')\bigr)\bigr|_{h' = h}
     \biggl(\frac{\r_n' (0|h, h)}{\r_n(0|h, h)}\biggr)^2 \epp
\end{multline}
Here we have implicitly assumed that $n \ne 0$. For $n = 0$ the
amplitude $A_0 \rightarrow 1$ as $h' \rightarrow h$, and the
product of the form factors goes to the square of the expectation
value of the current operator, which is vanishing, since the
current is odd under parity transformations, while the Hamiltonian
is even.

\subsection{Multiple-integral representations}
The explicit description of the general thermal form factors
of spin zero will be revealing more involved structures. We can
either describe them by means of multiple integrals or with
recourse to the so-called Fermionic basis. Both descriptions 
have their merits. As we shall see the multiple-integral
representations of the form factors are relatively simple
closed-form expressions. However, due to the implicit construction
of the involved integration contours, they are not efficient
for the actual computation of the form factors. For the
latter purpose the Fermionic basis is more appropriate.

We shall start our discussion with the multiple-integral
representations. They can be obtained by closely following
Appendix~A of the paper \cite{BoGo09}. The information
about the dominant state and the excited state composing
a specific form factor can be encoded in an auxiliary function
$\fa_n$ that satisfies a non-linear integral equation.
For the derivation of such integral equations in the context
of the dynamical case see \cite{Sakai07,GKKKS17}. The integral
equations can be written in several equivalent forms. Here
we choose a form in which the excitations are distinguished
by equivalence classes of contours ${\cal C}_n$.

Two functions, the bare energy
\begin{equation}
     \e_0 (\la) = h - \frac{4J (\D^2 - 1)}{\D - \cos(2 \la)} \epc    
\end{equation}
and the kernel function
\begin{equation} \label{kerf}
     K (\la) = \ctg(\la - \i \g) - \ctg(\la + \i \g) \epc
\end{equation}
are needed in the definition of the non-linear integral equation,
\begin{equation} \label{nlie}
     \ln \fa_n (\la|h) = - \frac{\e_0 (\la - \i \g/2)}{T}
        + \int_{{\cal C}_n} \frac{\rd \m}{2 \p \i} K(\la - \m) 
	     \ln_{{\cal C}_n} (1 + \fa_n) (\m|h) \epp
\end{equation}
The simple closed contours ${\cal C}_n$ are such that
$0 \in \Int {\cal C}_n$, $\la \pm \i \g \in \Ext {\cal C}_n$
if $\la \in \Int {\cal C}_n$, and
\begin{equation}
     \int_{{\cal C}_n} \rd \la \:
        \frac{\fa_n' (\la|h)}{1 + \fa_n (\la|h)} = 0 \epp
\end{equation}
The function $\ln_{{\cal C}_n} (1 + \fa_n)$ is the logarithm along
the contour ${\cal C}_n$. For its definition we fix a point
$\k_n \in {\cal C}_n$. For every $\la \in {\cal C}_n$ define
a partial contour ${\cal C}_{\k_n}^\la$ running counterclockwise
along ${\cal C}_n$ from $\k_n$ to $\la$. Then
\begin{equation}
     \ln_{{\cal C}_n} (1 + \fa_n) (\la|h) = 
        \int_{{\cal C}_{\k_n}^\la} \rd \m \:
	   \frac{\fa_n' (\m|h)}{1 + \fa_n (\m|h)}
	   + \ln \bigl(1 + \fa_n (\k_n|h)\bigr) \epc
\end{equation}
where the symbol `$\ln$' on the right hand side denotes the
principal branch of the logarithm. In the non-linear integral
equation (\ref{nlie}) for the auxiliary function the Trotter
limit $N \rightarrow \infty$ is already taken. For the finite
Trotter number version of this equations see \cite{GKKKS17}.

Two more important sequences of functions $G_n^{(\pm)}$ occurring
in the multiple-integral representation are defined as the
solutions of the linear integral equations
\begin{multline} \label{defg}
     G_n^{(\pm)} (\la, \x) = q^{\mp \a} \ctg(\la - \x + \i \g) 
        - \r_n^{\pm 1} (\x|h, h') \ctg(\la - \x) \\
	- \int_{{\cal C}_n^{(\pm)}} \rd m_n^{(\pm)} (\m) \:
	   K_{\mp \a} (\la - \m) G_n^{(\pm)} (\m, \x) \epp
\end{multline}
Here $\x \in \Int {\cal C}_n^{(\pm)}$,
\begin{equation}
     K_\a (\la) = q^{-\a} \ctg(\la - \i \g) - q^\a \ctg(\la + \i \g)
\end{equation}
is a deformed version of the kernel function (\ref{kerf}), and
the integration is performed with respect to the `measures'%
\footnote{Note that we have swapped the definitions of $\rd m^{(+)}$
and $\rd m^{(-)}$ as compared with our previous work \cite{DGKS16b}.}
\begin{equation}
     \rd m_n^{(+)} (\la) =
        \frac{\rd \la}
	     {2 \p \i \r_n (\la|h, h') \bigl(1 + \fa_0 (\la|h)\bigr)} \epc \qd
     \rd m_n^{(-)} (\la) =
        \frac{\rd \la \: \r_n (\la|h, h')}
	     {2 \p \i \bigl(1 + \fa_n (\la|h')\bigr)} \epp
\end{equation}
The contours ${\cal C}_n^{(\pm)}$ are deformations of the
contour ${\cal C}_n$ in such a way that the zeros of
$\r_n (\cdot|h,h')$ are excluded from ${\cal C}_n$ for
${\cal C}_n^{(+)}$, while the poles of $\r_n (\cdot|h,h')$
are excluded from ${\cal C}_n$ for ${\cal C}_n^{(-)}$.

In preparation of the following theorem we finally introduce
the short-hand notations
\begin{equation}
     \rd \overline{m}_n^{(+)} (\la) = \fa_0 (\la|h) \rd m_n^{(+)} (\la) \epc \qd
     \rd \overline{m}_n^{(-)} (\la) = \fa_n (\la|h') \rd m_n^{(-)} (\la) \epp
\end{equation}

\begin{theorem}
For all $\x_j \in \Int {\cal C}_n^{(\pm)}$, $j = 1, \dots, m$,
the form factors ${\cal F}_{n, m}^{(\pm)} (\xiv|h, h')$ of spin-zero
operators have the multiple-integral representations
\begin{multline} \label{ffmultint}
     {{\cal F}_{n;m}^{(\pm)}}^{\a_1 \dots \a_m}_{\be_1 \dots \be_m} (\xiv|h, h') = \\
        \biggl[ \prod_{j=1}^p
	        \int_{{\cal C}_n^{(\pm)}} \rd m_n^{(\pm)} (\la_j) \:
	        F_{x_j}^+ (\la_j) \biggr]
        \biggl[ \prod_{j = p + 1}^m
	        \int_{{\cal C}_n^{(\pm)}} \rd \overline{m}_n^{(\pm)} (\la_j) \:
	        F_{x_j}^- (\la_j) \biggr] \\ \times
    \frac{\det_m \bigl\{ - G_n^{(\pm)} (\la_j, \x_k)\bigr\}}
         {\prod_{1 \le j < k \le m} \sin(\la_j - \la_k + \i \g) \sin(\x_k - \x_j)}
\end{multline}
where
\begin{equation}
     F_x^\pm (\la) = \biggl[\prod_{k=1}^{x-1} \sin(\la - \x_k)\biggr]
                     \biggl[\prod_{k = x + 1}^m \sin(\la - \x_k \pm \i \g)\biggr] \epp
\end{equation}
The letter $p$ denotes the number of plusses in the sequence
$(\be_j)_{j=1}^m$ and the sequence $(x_j)_{j=1}^m$ is defined as
\begin{equation}
     x_j = \begin{cases}
              \e_j^+ & j = 1, \dots, p \\
	      \e_{m - j + 1}^- & j = p + 1, \dots, m
           \end{cases}
\end{equation}
with $\e_j^+$ being the position of the $j$th plus in 
$(\be_j)_{j=1}^m$, $\e_j^-$ that of the $j$th minus in
$(\a_j)_{j=1}^m$.
\end{theorem}

\begin{remark}
Specializing the multiple-integral representations
for $m = 1$, $\a_1 = \be_1 = +$, we see that
\begin{equation} \label{fpp1int}
     \tr\bigl\{e^+_+ {\cal F}_{n; 1}^{(\pm)} (\x|h, h')\bigr\} =
        - \int_{{\cal C}_n^{(\pm)}} \rd m_n^{(\pm)} (\la) \: G_n^{(\pm)} (\la,\x) \epp
\end{equation}
Using the technique developed in \cite{BoGo10} we directly reproduce
equation (\ref{fpp1}) from here. An obvious disadvantage of
(\ref{fpp1int}) is that the first order pole in $\a$ is not obvious in
this representation.
\end{remark}
\subsection{Factorization of the double integrals}
The double integrals for the two-site form factors, cf.\ the case
$m = 2$ in (\ref{ffmultint}), factorize by the same token as the
double integrals representing the generalized reduced density matrix
for two lattice sites \cite{BoGo09}. We can literally follow
Section~5 of \cite{BoGo09} to express all two-site form factors
in terms of only two functions defined by single integrals,
\begin{subequations}
\begin{align} \label{defphi}
     & \Phi_n^{(\pm)} (\x) =
        \int_{{\cal C}_n^{(\pm)}} \rd m_n^{(\pm)} (\la) \: G_n^{(\pm)} (\la,\x) \epc \\[1ex]
     & \Psi_n^{(\pm)} (\x_1,\x_2) =
        - \i \int_{{\cal C}_n^{(\pm)}} \rd m_n^{(\pm)} (\la) \:
	   \bigl(q^{\pm \a} \ctg(\la - \x_1 + \i \g) \notag \\ & \mspace{260.mu}
	         - \r_n^{\pm 1} (\x_1|h, h') \ctg(\la - \x_1)\bigr)
	   G_n^{(\pm)} (\la,\x_2) \epp
\end{align}
\end{subequations}
After some algebra we obtain
\begin{multline} \label{twositeff}
     {\cal F}_{n; 2}^{(\pm)} (\x_1, \x_2|h, h') \\ =
        \frac{1}{2(\z - \z^{-1})} \biggl\{
	- \frac{\z \Psi_n^{(\pm)} (\x_1,\x_2) - \z^{-1} \Psi_n^{(\pm)} (\x_2,\x_1)}
	       {q^{1 \pm \a} - q^{- 1 \mp \a}}
	     \begin{pmatrix}
	        -q & & & \\ & q & \z^{-1} & \\ & \z & q^{-1} & \\ & & & - q^{-1}
	     \end{pmatrix} \\
	+ \frac{\Psi_n^{(\pm)} (\x_1,\x_2) - \Psi_n^{(\pm)} (\x_2,\x_1)}
	       {q^{\pm \a} - q^{\mp \a}}
	     \begin{pmatrix}
	        - \z - \z^{-1} & & & \\ & \z + \z^{-1} & q + q^{-1} & \\
		& q + q^{-1} & \z + \z^{-1} & \\ & & & - \z - \z^{-1}
	     \end{pmatrix} \\
	+ \frac{\z^{-1} \Psi_n^{(\pm)} (\x_1,\x_2) - \z \Psi_n^{(\pm)} (\x_2,\x_1)}
	       {q^{1 \mp \a} - q^{- 1 \pm \a}}
	     \begin{pmatrix}
	        -q^{-1} & & & \\ & q^{-1} & \z & \\ & \z^{-1} & q & \\ & & & - q
	     \end{pmatrix} \biggr\} \\
        + \frac{q^{\pm \a} - q^{\mp \a}}{2(\z - \z^{-1})} \biggl\{\mspace{422.mu} \\[-2ex]
	  \frac{\z \Phi_n^{(\pm)} (\x_1) (\Phi_n^{(\pm)} (\x_2) + 1)
	        - \z^{-1} \Phi_n^{(\pm)} (\x_2) (\Phi_n^{(\pm)} (\x_1) + 1)}
	       {q^{1 \mp \a} - q^{- 1 \pm \a}}
	     \begin{pmatrix}
	        -q^{-1} & & & \\ & q^{-1} & \z & \\ & \z^{-1} & q & \\ & & & - q
	     \end{pmatrix} \\
	+ \frac{\z^{-1} \Phi_n^{(\pm)} (\x_1) (\Phi_n^{(\pm)} (\x_2) + 1)
	        - \z \Phi_n^{(\pm)} (\x_2) (\Phi_n^{(\pm)} (\x_1) + 1)}
	       {q^{1 \pm \a} - q^{- 1 \mp \a}}
	     \begin{pmatrix}
	        -q & & & \\ & q & \z^{-1} & \\ & \z & q^{-1} & \\ & & & - q^{-1}
	     \end{pmatrix}
	     \biggr\} \\
	+ \frac{\Phi_n^{(\pm)} (\x_1) - \Phi_n^{(\pm)} (\x_2)}
	       {2(\z - \z^{-1})}
	     \begin{pmatrix}
	        \z - \z^{-1} & & & \\ & \z^{-1} - \z & q - q^{-1} & \\
		& q^{-1} - q & \z - \z^{-1} & \\ & & & \z^{-1} - \z
	     \end{pmatrix}\\ +
	     \begin{pmatrix}
	        - \Phi_n^{(\pm)} (\x_1) & & & \\ & 0 & & \\
		& & 0 & \\ & & & 1 + \Phi_n^{(\pm)} (\x_1)
	     \end{pmatrix} \epc
\end{multline}
where $\z = \re^{\i(\x_1 - \x_2)}$.
\subsection{Connection with the Fermionic basis}
Using the result of the previous subsection we can now establish a
relation with the Fermionic basis introduced in \cite{BJMST08a,%
BJMS09a}, in particular with a remarkable theorem (the JMS theorem) 
proved by Jimbo, Miwa and Smirnov in \cite{JMS08}. In its original
formulation it states that expectation values calculated with
${\cal F}_{0;m}^{(+)}$, which can be interpreted as a generalized
reduced density matrix, factorize and can all be expressed in terms
of only two transcendental functions, the function $\r_0$ entering
the reduction relations (\ref{defrhogen}) and another function $\om$
which in \cite{JMS08} was defined as the expectation value of a product
of two creation operators and was represented by a determinant formula.
The point of view taken in \cite{JMS08} is slightly different from ours
here in that the lattice used in \cite{JMS08} is homogeneous in
`horizontal direction'. In the inhomogeneous case, we can also follow
sections 5.1 and 5.3 of \cite{BJMST08a}, where ground state expectation
values were considered and $\om$ was expressed as\footnote{More
precisely this function was denoted $(\om_0 - \om)(\z_1/\z_2, \a)$
in \cite{BJMST08a}.}
\begin{equation} \label{defom}
     \om (\x_1, \x_2)
        = - \bigl\langle
	     \cv^*_{[1,2]} (\z_2,\a) \bv^*_{[1,2]}(\z_1,\a - 1) (1)
	    \bigr\rangle
\end{equation}
with $\z_j = \re^{\i \x_j}$.

As is clear from the derivation and as was emphasized in \cite{JMS21},
the validity of the JMS theorem extends to a much larger class
of expectation values. It remains true for arbitrary inhomogeneities
in the horizontal spaces and for arbitrary spin-zero states. Hence,
it also applies to the present case with $\om$ being replaced by
\begin{equation} \label{defomff}
     \om_n^{(\pm)} (\x_1, \x_2|h, h')
        = - \tr\bigl\{{\cal F}_{n; 2}^{(\pm)} (\x_1, \x_2|h, h')
	     \cv^*_{[1,2]} (\z_2,\mp \a) \bv^*_{[1,2]}(\z_1,\mp \a - 1) (1)\bigr\}
\end{equation}
and `expectation values' calculated with ${\cal F}_{n; m}^{(\pm)}$
instead of ${\cal F}_{0; m}^{(+)}$.

The construction of the operators $\bv^*_{[1,2]}$ and $\cv^*_{[1,2]}$ is
explained in \cite{BJMST08a}. For the product needed in (\ref{defomff})
we find the explicit expression
\begin{align} \label{c2b11}
     \z^{\a} & \cv^*_{[1,2]} (\z_2, - \a) \bv^*_{[1,2]} (\z_1,- \a - 1) (1) =
        \notag \\[1ex] &
        \biggl( \frac{q^{\a + 1} \z^{-1}}{q^{-1} \z - q \z^{-1}} -
                \frac{q^{- 1 - \a} \z^{-1}}{q \z - q^{-1} \z^{-1}} +
                \frac{q^\a - q^{- \a}}{2} \biggr) \s^z \otimes \s^z
        \notag \\[1ex] & +
        \frac{q^\a - q^{- \a}}{2}
        \biggl( \frac{q \z^{-1}}{q^{-1} \z - q \z^{-1}} -
                \frac{q^{-1} \z^{-1}}{q \z - q^{-1} \z^{-1}} \biggr)
                \bigl( I_2 \otimes \s^z - \s^z \otimes I_2 \bigr)
        \notag \\[1ex] & +
        2 \biggl( \frac{q^\a}{q^{-1} \z - q \z^{-1}} -
                  \frac{q^{- \a}}{q \z - q^{-1} \z^{-1}} \biggr)
                  \bigl( \s^+ \otimes \s^- + \s^- \otimes \s^+ \bigr)
        \notag \\[1ex] & +
        (q^\a - q^{- \a})
        \biggl( \frac{1}{q^{-1} \z - q \z^{-1}} +
                \frac{1}{q \z - q^{-1} \z^{-1}} \biggr)
                \bigl( \s^+ \otimes \s^- - \s^- \otimes \s^+ \bigr) \epc
\end{align}
where $\z = \z_1/\z_2$.

Inserting this into (\ref{defomff}) and calculating the average
with the factorized two-site form factors (\ref{twositeff}) of
the previous section we obtain the following
\begin{lemma}
The functions that, together with $\r_n$ defined in (\ref{defrho}),
determine all form factors ${\cal F}_{n; m}^{(\pm)} (0|h, h')$
through the JMS theorem \cite{JMS08} are
\begin{multline} \label{ompsi}
     \om_n^{(\pm)} (\x_1, \x_2|h, h') = \\
        2 \z^{\mp \a} \Ps_n^{(\pm)} (\x_1, \x_2) + \D _\z \ps(\z, \mp \a)
                       + 2 \bigl( \r_n^{\pm 1} (\x_1|h,h') -
		                  \r_n^{\pm 1} (\x_2|h,h') \bigr) \ps(\z, \mp \a) \epp
\end{multline}
Here we adopted the notation from \cite{BJMST08a},
\begin{equation} \label{defpsi}
     \ps(\z, \a) = \frac{\z^\a (\z^2 + 1)}{2(\z^2 - 1)} \epc
\end{equation}
and $\D_\z$ is the difference operator whose action on a function $f$
is defined by $\D_\z f(\z) = f(q \z) - f(q^{-1} \z)$.
\end{lemma}
For more details about how the JMS theorem determines
the correlation functions (resp.\ the form factors) and some
examples of short-range correlation functions see
\cite{GKW21,MiSm19}.

\subsection{Example of the energy density correlations}
The form factors of the energy density operator
\[
     {\cal E}/J = 2 \bigl(\s^- \otimes \s^+ + \s^+ \otimes \s^-\bigr)
                  + \tst{\2} (q + q^{-1}) \s^z \otimes \s^z
\]
are the simplest example of form factors that involve $\om_n^{(\pm)}$.
Using (\ref{twositeff}) and (\ref{ompsi}) we obtain them in the form
\begin{multline}
     \lim_{h' \rightarrow h} \lim_{\x_j, \z_k \rightarrow 0} A_n (h, h')
     \tr \bigl\{{\cal E}_{1,2} {\cal F}_{n; 2}^{(-)} (\x_1, \x_2|h, h')\bigr\}
     \tr \bigl\{{\cal E}_{1,2} {\cal F}_{n; 2}^{(+)} (\z_1, \z_2|h, h')\bigr\} = \\
     \frac{J^2 \sh^2 (\g)}{2} \bigl(\6_{h'}^2 A_n (h, h')\bigr)\bigr|_{h' = h}
     \res_{h' = h} \om_n^{(+)} (0,0|h,h') \res_{h' = h} \om_n^{(-)} (0,0|h,h') \epc
\end{multline}
where we have assumed that $n \ne 0$.

Again the case $n = 0$ must be considered separately, since
$\lim_{h' \rightarrow h} A_0 (h, h') = 1$ and the functions
$\om_0^{(\pm)} (0,0|h,h')$ are regular at $h = h'$. Using
(\ref{defreddens}) we conclude that
\begin{multline} \label{ffnzeroenergy}
     \lim_{h' \rightarrow h} \lim_{\x_j, \z_k \rightarrow 0} A_0 (h, h')
     \tr \bigl\{{\cal E}_{1,2} {\cal F}_{0; 2}^{(-)} (\x_1, \x_2|h, h')\bigr\}
     \tr \bigl\{{\cal E}_{1,2} {\cal F}_{0; 2}^{(+)} (\z_1, \z_2|h, h')\bigr\} = \\
     \bigl(\tr \bigl\{{\cal E}_{1,2} {\cal D}_2 (0, 0|h, h)\bigr\}\bigr)^2
     = \lim_{h' \rightarrow h} J^2 \sh^2 (\g) \om_0^{(+)} (0,0|h,h') \om_0^{(-)} (0,0|h,h') \epp
\end{multline}
\begin{remark}
The latter formula provides and opportunity for checking the
consistency of our result. On the one hand we obtain from
(\ref{ompsi}) that
\begin{multline} \label{omzeroexpl}
     - J \sh(\g) \lim_{h' \rightarrow h} \om_0^{(\pm)} (0,0|h,h') = \\
        \frac{J \sh(\g)}{\p}
	   \int_{{\cal C}_0} \frac{\rd \la}{1 + \fa_0 (\la|h)}
	   \bigl(\ctg(\la + \i \g) - \ctg(\la)\bigr)
	   G_0^{(\pm)} (\la, 0)\bigr|_{\a = 0} + J \D \epp
\end{multline}
On the other hand the thermal expectation value
$\bigl\<{\cal E}_{1,2}\bigr\>_T = \tr \bigl\{{\cal E}_{1,2}
{\cal D}_2 (0, 0|h, h)\bigr\}$ can be calculated as
\begin{equation}
     \bigl\<{\cal E}_{1,2}\bigr\>_T = J \6_{J/T} \frac{f(T, h)}{T} + J \D \epc
\end{equation}
where $f$ is the free energy per lattice site which has the
representation
\begin{multline}
     f(T,h) = - T \ln \bigl(\La_0 (0|h)\bigr) = \\
        \frac{h}{2} - T \int_{{\cal C}_0} \frac{\rd \la}{2\p \i} 
	                   \bigl(\ctg(\la) - \ctg(\la + \i \g)\bigr)
			   \ln_{{\cal C}_0} (1 + 1/\fa_0) (\la|h) \epp
\end{multline}
From the nonlinear integral equation (\ref{nlie}) and from
(\ref{defg}) we can infer that
\begin{equation} \label{gfroma}
     G_0^{(\pm)} (\la, 0)\bigr|_{\a = 0} =
        - \frac{1}{2 \i \sh(\g)} \frac{\6_{J/T} \fa_0 (\la|h)}{\fa_0 (\la|h)} \epp
\end{equation}
Then (\ref{omzeroexpl})-(\ref{gfroma}) imply that
\begin{equation}
     - J \sh(\g) \lim_{h' \rightarrow h} \om_0^{(\pm)} (0,0|h,h') =
        \bigl\<{\cal E}_{1,2}\bigr\>_T \epc
\end{equation}
and we have obtained a direct verification of (\ref{ffnzeroenergy}).
\end{remark}
\begin{remark}
Another closely related operator of some interest in applications
\cite{KlSa02,SaKl03} is the heat current operator
\begin{equation}
     {\cal J}^{\cal E} = [{\cal E} \otimes I_2, I_2 \otimes {\cal E}] \epc
\end{equation}
where $I_2$ is the $2 \times 2$ unit matrix. This operator is an
example of a relevant local spin-zero operator that acts non-trivially
on three consecutive lattice sites. Examples of local spin-zero
operators with even larger `support' are densities of the higher
conserved charges generated by the transfer matrix $t_\perp (\la)$
and the corresponding currents \cite{Pozsgay20}.
\end{remark}
\subsection{\boldmath Homogeneous limit and limit $h' \rightarrow h$
in the general case}
As it is clear from (\ref{tffszero}) and from the examples considered
above the calculation of any physical two-point function requires to
take the homogeneous limit $\x_j\rightarrow 0$ and the limit
$h' \rightarrow h$. Instead of $h'$ we may also consider $\g \a =
(h - h')/2T$ as the independent regularization parameter. This
choice is more natural for taking the limit $T \rightarrow 0_+$ 
subsequently as we shall see in the next section. Let $X \in \End
\bigl(({\mathbb C}^2)^{\otimes \ell}\bigr)$, $Y \in \End
\bigl(({\mathbb C}^2)^{\otimes r}\bigr)$ be two arbitrary spin-zero
operators. Then (\ref{tffszero}) implies that
\begin{multline} \label{tffshomouni}
     \bigl\< X_{\llbracket 1, \ell\rrbracket} (t)
             Y_{\llbracket 1 + m, r + m\rrbracket} \bigr\>_T -
     \bigl\< X_{\llbracket 1, \ell\rrbracket}\bigr\>_T
     \bigl\<Y_{\llbracket 1, r\rrbracket} \bigr\>_T = \\[1ex]
	\lim_{N \rightarrow \infty} \;
     \sum_{n > 0} \tst{\2} \6_{\g \a}^2 A_n (h, h')\bigr|_{\g \a = 0}\,
     F_{X; n}^{(-)} F_{Y; n}^{(+)}\, \r_n (0|h, h)^m
     \biggl(
     \frac{\r_n \bigl(- \frac{\i t}{\k N}\big|h, h)}
          {\r_n \bigl(\frac{\i t}{\k N}\big|h, h)}
	  \biggr)^\frac{N}{2} \epc
\end{multline}
where
\begin{subequations}
\begin{align}
     F_{X; n}^{(-)} =
        \lim_{\g \a \rightarrow 0} \g \a \lim_{\x_j \rightarrow 0}
	\tr\bigl\{X {\cal F}_{n; \ell}^{(-)} (\xiv|h, h')\bigr\} \epc \\[1ex]
     F_{Y; n}^{(+)} =
        \lim_{\g \a \rightarrow 0} \g \a \lim_{\x_j \rightarrow 0}
	\tr\bigl\{Y {\cal F}_{n; r}^{(+)} (\xiv|h, h')\bigr\} \epp
\end{align}
\end{subequations}
The latter two functions are determined by the Fermionic basis,
while the first function under the sum in (\ref{tffshomouni}),
$\2 \6_{\g \a}^2 A_n (h, h')\bigr|_{\g \a = 0}$, is the same for
all spin-zero operators. In previous work we obtained explicit
expressions for the eigenvalue ratios \cite{DGKS15b} and for
the functions $\2 \6_{\g \a}^2 A_n (h, h')\bigr|_{\g \a = 0}$
\cite{BGKS21a,BGKSS21} for the XXZ chain in the antiferromagnetic
massive regime in the limit $T \rightarrow 0_+$.

\section{The antiferromagnetic massive regime at zero temperature}
\label{sec:massive_low_T_xxz}
The sum in (\ref{tffshomouni}) is a spectral decomposition of a
dynamical two-point function over the spectrum of the dynamical
quantum transfer matrix. Any further analysis of the sum and of
its Trotter limit requires detailed knowledge of this spectrum.
It is best understood in the cases of high \cite{GGKS20} and low
\cite{DGKS15b,FGK23pp} temperature. In previous work on the
two-point function $\<\s_1^z (t) \s_{m+1}^z\>_T$ of the local
magnetization of the XXZ chain in its antiferromagnetic massive
regime in the limit $T \rightarrow 0_+$ we were able to rewrite
the sum as a sum over multiple integrals with explicit integrands
\cite{BGKS21a}. In that work we used a generating function of
the two-point function rather than the properly normalized form
factors. For this reason the analysis was restricted to
$\<\s_1^z (t) \s_{m+1}^z\>_T$. We shall now see that the previous
work can be extended to the analysis of arbitrary local operators
of spin zero as long as we stay in the antiferromagnetic massive
regime at $T \rightarrow 0_+$.
\subsection{Low-temperature spectrum of the quantum transfer matrix}
In this case the full spectrum of the quantum transfer matrix
was obtained in \cite{DGKS15b}. In the low-$T$ limit the $n$th
excitation is parameterised by two sets of particle and hole
rapidities, ${\cal Y}_n$ and ${\cal X}_n$ and by an `index'
$k \in \{0, 1\}$ that distinguishes two separate `ground-state'
sectors.

For a quantitative description of the sets ${\cal Y}_n$, ${\cal X}_n$
of particle and hole parameters and of the eigenvalue ratios in
the low-temperature limit we need to introduce the dressed momentum
$p$ and the dressed energy $\e$ as well as the dressed phase function
$\ph$ that characterizes the two-particle scattering. In the
present case these can be written in terms of special functions
from the $q$-gamma family, namely in terms of the $q$-gamma function
$\G_{q^4}$ and of the closely related Jacobi theta functions
$\dh_j (\cdot|q)$, $j = 1, \dots, 4$, and their specializations
$\dh_j = \dh_j (0|q)$ (see e.g.\ \cite{BGKS21a} for the definitions).

Dressed momentum and dressed energy as functions of the
rapidity $\la$ are defined as
\begin{align} \label{ptheta4}
      p(\la) & = \frac{\p}{2} + \la
             - \i \ln \biggl(
	       \frac{\dh_4 (\la + \i \g/2| q^2)}{\dh_4 (\la - \i \g/2| q^2)}
	       \biggr) \epc \\[1ex] \label{dressede}
     \e(\la|h) & = \frac h 2 - 2 J \sh (\g) \dh_3 \dh_4
                           \frac{\dh_3 (\la|q)}{\dh_4 (\la|q)} \epp
\end{align}
Note that the dressed energy depends parametrically on the magnetic field
$h$. The condition $\e(\frac{\p}{2}|h_\ell) = 0$ determines the lower
critical field
\begin{equation} \label{hlow}
     h_\ell = 4 J \sh(\g) \dh_4^2 \epc
\end{equation}
i.e.\ the boundary of the antiferromagnetic massive regime $|h| < h_\ell$
at $\D > 1$.

The dressed phase is the function
\begin{equation} \label{dressedphase}
     \ph(\la_1, \la_2) = \i \Bigl( \frac \p 2 + \la_{12} \Bigr)
        + \ln \Biggl\{ \frac{\G_{q^4} \bigl(1 + \frac{\i \la_{12}}{2\g}\bigr)
	                     \G_{q^4} \bigl(\2 - \frac{\i \la_{12}}{2\g}\bigr)}
		            {\G_{q^4} \bigl(1 - \frac{\i \la_{12}}{2\g}\bigr)
			     \G_{q^4} \bigl(\2 + \frac{\i \la_{12}}{2\g}\bigr)}
			     \Biggr\} \epc
\end{equation}
where $\la_{12} = \la_1 - \la_2$, $|\Im \la_2| < \g$.

For differences of sums and ratios of products we introduce
a convenient short-hand notation. Given two finite sets ${\cal X},
{\cal Y} \subset {\mathbb C}$
%associate a map $d: {\cal X} \cup {\cal Y} \mapsto {\mathbb C}$,
%\begin{equation}
%     d_{{\cal X}, {\cal Y}} (z) =
%        \begin{cases}
%	   1 & z \in {\cal X} \setminus ({\cal X} \cap {\cal Y}) \\
%	   0 & z \in {\cal X} \cap {\cal Y} \\
%	   - 1 & z \in {\cal Y} \setminus ({\cal X} \cap {\cal Y})
%        \end{cases}
%\end{equation}
we associate with any function $f: {\mathbb C} \rightarrow {\mathbb C}$
the expressions 
\begin{equation} \label{defsetsumprod}
     \sum_{\la \in {\cal X} \ominus {\cal Y}} f(\la)
        = \sum_{\la \in {\cal X}} f(\la)
          - \sum_{\la \in {\cal Y}} f(\la) \epc \qd
     \prod_{\la \in {\cal X} \ominus {\cal Y}} f(\la)
        = \frac{\prod_{\la \in {\cal X}} f(\la)}{\prod_{\la \in {\cal Y}} f(\la)} \epp
\end{equation}
In particular, the shift function $F$ with parameters $\cal X$,
$\cal Y$ is defined as
\begin{equation}
    F(\la|{\cal X}, {\cal Y}) = \frac{1}{2 \p \i}
        \sum_{\m \in {\cal Y} \ominus {\cal X}} \ph(\la, \m) \epp
\end{equation}

%\begin{figure}
%\begin{center}
%\includegraphics[width=.70\textwidth]{reepszero_crit_noncrit}
%\end{center}
%\caption{\label{fig:regimes} The curves $\Re \e (\la|h) = 0$ for various
%values of the magnetic field. Here $\D = 1.7$, $h_\ell/J = 0.76$.
%The values of the magnetic field decrease proceeding from the inner
%to the outer curve: $h/h_\ell = 1.34, 1, 2/3, 1/3, 0$. In the
%critical regime $h_\ell < h < h_u = 4J (1 + \D)$ the curves are
%closed. In this regime $\e$ is not given by (\ref{dressede}), but
%is rather defined as a solution of a linear integral equation.
%At the lower critical field $h = h_\ell$ the closed curves develop
%two cusps, and a gap opens for $0 < h < h_\ell$, which is the
%parameter regime considered in this work.}
%\end{figure}

Our analysis in \cite{BGKS21a} was based on the following
\begin{conjecture}
\label{con:nostrings}
{\bf \boldmath Low-$T$ Bethe root patterns at $0 < h < h_\ell$ \cite{DGKS15b}.}
\begin{enumerate}
\item
All excitations of the quantum transfer matrix at low $T$ and
large enough Trotter number $N$ can be parameterised by an even
number of complex parameters located inside the strip $|\Im \la| < \g/2$
and by an index $\fe \in {\mathbb Z}\slash 2 {\mathbb Z} $.
Referring to \cite{DGKS15b} we call the parameters in the upper
half plane particles, the parameters in the lower half plane holes.
We shall denote the set of particles by $\cal Y$, the set of
holes by $\cal X$.
\item
For pseudo spin zero states the cardinality of the sets of particles
and holes must be equal, $\card {\cal X} = \card {\cal Y} = \nex$.
\item
Up to corrections of the form $T^\infty + a_N (T)$, where
$\lim_{N \rightarrow \infty} a_N (T) = 0$ for every fixed
$T > 0$, the particles and holes are determined by the low-$T$
higher-level Bethe Ansatz equations
\begin{subequations}
\label{hlbaes}
\begin{align}
     \e(y|h) & = 2 \p \i T \bigl(\ell_y + \fe/2 + F(y|{\cal X}, {\cal Y})\bigr) \epc
			     && \forall\ y \in {\cal Y} \epc \\[1ex]
     \e(x|h) & = 2 \p \i T \bigl(m_x + \fe/2 + F(x|{\cal X}, {\cal Y})\bigr) \epc
			     && \forall\ x \in {\cal X} \epc
\end{align}
\end{subequations}
where $\ell_y, m_x \in {\mathbb Z}$ and where the $\ell_y$ and the
$m_x$ are mutually distinct.
\item
For a solution ${\cal X}_n$, ${\cal Y}_n$ of (\ref{hlbaes}) with index
$\fe = \fe_n$ we denote $F (\la|{\cal X}_n, {\cal Y}_n) = F_n (\la)$. The
auxiliary function corresponding to this solution is
\begin{equation} \label{anlowt}
     \fa_n (\la|h) = (-1)^{\fe_n} \re^{- \e(\la - \i \g/2|h)/T + 2\p\i F_n (\la - \i \g/2)}
                      \bigl(1 + {\cal O} \bigl(T^\infty + a_N (T)\bigr)\bigr)
\end{equation}
uniformly for $0 < |\Im \la| < \g$, away from the points $0$, $\i \g$.
\end{enumerate}
\end{conjecture}

We would like to emphasize that the above conjecture is formulated
for small $T$ and large enough Trotter number~$N$. The finite Trotter
number corrections to the higher-level Bethe equations and to the
auxiliary function $\fa_n$ are symbolized by ${\cal O} (T^\infty + a_N (T))$,
where, as explained in point (iii), $\lim_{N \rightarrow \infty} a_N (T) = 0$
for every fixed $T > 0$. In the following we shall keep this notation. It is
also implying that the limit $N \rightarrow \infty$ has to be taken before
the limit $T \rightarrow 0_+$. If the limits are taken in this order, the
higher-level Bethe Ansatz equations (\ref{hlbaes}) decouple,
$\i \p \ell_y T$ and $\i \p m_x T$ turn into independent continuous
variables, and the particles and holes become free parameters on the curves
\begin{equation} \label{defbpm}
     {\cal B}_\pm = \bigl\{ \la \in {\mathbb C} \big|
                            \Re \e (\la|h) = 0, - \p/2 \le \Re \la \le \p/2,
			    0 < \pm \Im \la < \g \bigr\} \epp
\end{equation}

The main result of our work \cite{DGKS15b} is an explicit formula
for all eigenvalue ratios $\r_n$ in the low-temperature regime.
According to the above conjecture they must be parameterised by
solutions ${\cal X}_n$, ${\cal Y}_n$ of the higher-level Bethe Ansatz
equations (\ref{hlbaes}). Let $k = \fe_n - \fe_0$ (effectively
meaning that $k \in \{0, 1\}$). Then, for $|\Im \la| < \g/2$, the
corresponding eigenvalue ratios at finite magnetic field
\cite{DGKS15b,DGKS16b} can be expressed as
\begin{equation} \label{evarat}
     \r_n (\la|h, h') = w(\la - \i \g/2|{\cal X}_n, {\cal Y}_n|k)
		    \bigl(1 + {\cal O} \bigl(T^\infty + a_N (T)\bigr)\bigr) \epc
\end{equation}
where
\begin{equation}
     w(\la|{\cal U}, {\cal V}|k) = (-1)^k 
        \exp \biggl\{\i \mspace{-8.mu}
	       \sum_{z \in {\cal U} \ominus {\cal V}}
	            \mspace{-7.mu} p(\la - z + \i \g/2)\biggr\}
	= (-1)^{k} \mspace{-8.mu} \prod_{z \in {\cal U} \ominus {\cal V}}
	           \frac{\dh_1(\la - z| q^2)}{\dh_4(\la - z| q^2)}
\end{equation}
and where we assume that $\card {\cal U} - \card {\cal V}$ is even.
Note that the value of the magnetic field $h'$ enters here through
the particle and hole parameters $z$.

From equation (\ref{evarat}) we obtain the eigenvalue ratios entering
the form-factor series (\ref{tffshomouni}),
\begin{equation} \label{evaratzero}
     \r_n (0|h, h') =
        (-1)^k \exp \biggl\{\i 
	       \mspace{-8.mu}
	        \sum_{z \in {\cal Y}_n \ominus {\cal X}_n}
	       \mspace{-15.mu} p(z) \biggr\}
	       \bigl(1 + {\cal O} \bigl(T^\infty + a_N (T)\bigr)\bigr)
\end{equation}
and
\begin{multline}
        \biggl(\frac{\r_n (- \i t/\ks N|h,h')}
	            {\r_n (\i t/\ks N|h,h')}\biggr)^\frac{N}{2} \\ =
        \exp \bigl\{- (\i t/\ks) \6_\la \ln w(\la|{\cal X}_n, {\cal Y}_n|k)
	            \bigr\}\Bigr|_{\la = - \i \g/2}
		    \bigl(1 + {\cal O} \bigl(T^\infty + a_N (T)\bigr)\bigr) \\[1ex]
     = \exp \biggl\{- \i t \mspace{-14.mu}
                     \sum_{z \in {\cal Y}_n \ominus {\cal X}_n}
	            \mspace{-16.mu} \e(z|h) \biggr\}
		    \bigl(1 + {\cal O} \bigl(T^\infty + a_N (T)\bigr)\bigr) \epp
\end{multline}
\subsection{Low-temperature limit of the universal amplitude}
An explicit expression for the low-temperature limit of the
`universal amplitude' $A_n (h, h')$, was obtained in
\cite{BGKS21a,BGKSS21}. It involves again only well-known special
functions of the $q$-gamma and $q$-hypergeometric families.

We fix two sets ${\cal U} \subset {\mathbb H}_-$ and ${\cal V}
\subset {\mathbb H}_+$ of equal finite cardinality, $\card {\cal U}
= \card {\cal V} = \ell$. Given these sets we define
\begin{equation}
     \Si^{(\a)} ({\cal U}, {\cal V}|k) =
        - \frac{\p k}{2} + \frac{\i \a \g}{2}
	- \2 \sum_{\la \in {\cal U} \ominus {\cal V}} \la
\end{equation}
and certain products involving the $q$-gamma functions $\G_{q^4}$ and
the $q$-Barnes functions $G_{q^4}$,
\begin{subequations}
\begin{align} \label{defphpm}
     \Ph^{(\pm)} (\la) & =
        \re^{\pm \i \Si^{(\a)} ({\cal U}, {\cal V}|k)}
        \prod_{\m \in {\cal U} \ominus {\cal V}}
        \tst{\G_{q^4} \bigl(\2 \pm \frac{\la - \m}{2 \i \g}\bigr)
             \G_{q^4} \bigl(1 \mp \frac{\la - \m}{2 \i \g}\bigr)} \epc \\[1ex]
     \Xi ({\cal U}, {\cal V}) & = \prod_{\la, \m \in {\cal U} \ominus {\cal V}}
        \frac{\G_{q^4} \bigl(\2 + \frac{\la}{2 \i \g}\bigr)
	      G_{q^4}^2 \bigl(1 + \frac{\la}{2 \i \g}\bigr)}
	     {\G_{q^4} \bigl(1 + \frac{\la}{2 \i \g}\bigr)
	      G_{q^4}^2 \bigl(\2 + \frac{\la}{2 \i \g}\bigr)} \epp
\end{align}
\end{subequations}

With $u_j \in {\cal U}$ and $v_k \in {\cal V}$ we associate `multiplicative
spectral parameters' $H_j = \re^{2 \i u_j}$, $P_k = \re^{2 \i v_k}$. In
terms of these we can define the following special basic hypergeometric series
(cf.\ \cite{GaRa04}),
\begin{subequations}
\begin{align}
     \PH_1 (P_k, \a) & = \:
        _{2 \ell} \PH_{2\ell - 1}
           \Biggl(\begin{array}{@{}r} q^{-2},
                     \{q^2 \frac{P_k}{P_m}\}_{m \ne k}^\ell,
                     \{\frac{P_k}{H_m}\}_{m=1}^\ell \\[.5ex]
                     \{\frac{P_k}{P_m}\}_{m \ne k}^\ell,
                     \{q^2 \frac{P_k}{H_m}\}_{m=1}^\ell
                  \end{array}; q^4, q^{4 + 2 \a}
           \Biggr), \\[1ex]
     \PH_2 (P_k, P_j, \a) & = \:
        _{2 \ell} \PH_{2\ell - 1}
           \Biggl(\begin{array}{@{}r} q^6, q^2 \frac{P_j}{P_k},
                     \{q^6 \frac{P_j}{P_m}\}_{m \ne k, j}^\ell,
                     \{q^4 \frac{P_j}{H_m}\}_{m=1}^\ell \\[.5ex]
                     q^8 \frac{P_j}{P_k},
                     \{q^4 \frac{P_j}{P_m}\}_{m \ne k, j}^\ell,
                     \{q^6 \frac{P_j}{H_m}\}_{m=1}^\ell
                  \end{array}; q^4, q^{4 + 2 \a}
           \Biggr) \epp
\end{align}
\end{subequations}
We further define
\begin{equation}
     \Ps_2 (P_k, P_j, \a) = q^{2 \a} r_\ell (P_k, P_j) \PH_2 (P_k, P_j, \a) \epc
\end{equation}
where
\begin{equation}
     r_\ell (P_k, P_j) = \frac{q^2 (1 - q^2)^2 \frac{P_j}{P_k}}
                              {(1 - \frac{P_j}{P_k})(1 - q^4 \frac{P_j}{P_k})}
        \Biggl[\prod_{\substack{m=1 \\ m \ne j, k}}^\ell
               \frac{1 - q^2 \frac{P_j}{P_m}}{1 - \frac{P_j}{P_m}}\Biggr]
        \Biggl[\prod_{m=1}^\ell \frac{1 - \frac{P_j}{H_m}}
	                             {1 - q^2 \frac{P_j}{H_m}}\Biggr] \epc
\end{equation}
and introduce a `conjugation' $\overline f (H_j, P_k, q^\a)
= f(1/H_j, 1/P_k, q^{- \a})$.

The above definitions allow us to introduce a matrix
${\cal M}_\a ({\cal U}, {\cal V}|k)$ with elements
\begin{multline}
     {{\cal M}_\a}_j^i ({\cal U}, {\cal V}|k) =
        \de_j^i \biggl[\overline{\PH}_1 (P_j, \a) -
                        \frac{\Ph^{(-)} (v_j)}{\Ph^{(+)} (v_j)} \PH_1 (P_j, \a)\biggr] \\
        - (1 - \de_j^i )\biggl[\overline{\Ps}_2 (P_j, P_i, \a) -
                               \frac{\Ph^{(-)} (v_i)}{\Ph^{(+)} (v_i)}
                               \Ps_2 (P_j, P_i, \a)\biggr] \epp
\end{multline}
By $\hat {\cal M}_\a ({\cal U}, {\cal V}|k)$ we denote the matrix
obtained from ${\cal M}_\a ({\cal U}, {\cal V}|k)$ upon replacing
$u_j \leftrightharpoons - v_j$ for $j = 1, \dots, \ell$.

Next we define what we call the `universal weight function'
\begin{multline} \label{weightfun}
     {\cal W}_\a^{(2 \ell)} ({\cal U}, {\cal V}|k) = -
        \biggl(\frac{\dh_1'}
	            {2 \dh_1 \bigl(\Si^{(\a)} ({\cal U}, {\cal V}|k)\big|q \bigr)}\biggr)^2
	\det_\ell \biggl(\frac{1}{\sin(u_j - v_k)}\biggr)^2
	   \\[1ex] \times \Xi ({\cal U}, {\cal V})
	\det_\ell \{{\cal M}_\a ({\cal U}, {\cal V}|k)\}
	\det_\ell \{\hat {\cal M}_\a ({\cal U}, {\cal V}|k)\}
	\epp
\end{multline}
This function determines one factor of the universal amplitude in
the low-temperature limit.

Another factor is given in terms of the auxiliary function $\fa_n$.
An off-shell version of this function was defined in \cite{DGKS15b,BGKS21a}.
If $T$ is small enough and $N$ large enough, we may set
\begin{equation} \label{offshelllowt}
     \fa^\pm (\la|{\cal U}, {\cal V}, h) =
        (-1)^\fe \re^{\mp \frac{\e(\la|h)}{T} \pm 2\p \i F(\la|{\cal U}, {\cal V})} \epp
\end{equation}
For arbitrary finite temperatures we may define non-linear integral
equations whose solutions converge to (\ref{offshelllowt}) in the
low-$T$ limit (see Sec.~3.3 of \cite{BGKS21a}). For $T$ small enough
and $N$ large enough it is clear from (\ref{anlowt}) that
\begin{equation}
     \fa_n (\la|h') = \fa^+ (\la - \i \g/2|{\cal X}_n, {\cal Y}_n, h')
                    = 1/\fa^- (\la - \i \g/2|{\cal X}_n, {\cal Y}_n, h')
\end{equation}
up to multiplicative corrections of the form $1 + {\cal O}\bigl(
T^\infty + a_N (T)\bigr)$ uniformly for $0 < \Im \la < \g$ away
from $0$, $\i \g$.

We set
\begin{equation}
     {\mathbb J} ({\cal U}, {\cal V}, h') =
	\begin{pmatrix}
	   \6_{u_k} \fa^- (u_j |{\cal U},{\cal V},h') &
	   \6_{v_k} \fa^- (u_j |{\cal U},{\cal V},h') \\
	   \6_{u_k} \fa^+ (v_j |{\cal U},{\cal V},h') &
	   \6_{v_k} \fa^+ (v_j |{\cal U},{\cal V},h')
	\end{pmatrix} \epp
\end{equation}
At this point everything is prepared to present a result for
the universal amplitude in the low-$T$ limit that follows
from the formulae in \cite{BGKS21a,BGKSS21} upon elementary
manipulations.
\begin{lemma}
Amplitudes in the low-$T$ limit \cite{BGKS21a,BGKSS21}. For the
XXZ chain in the antiferromagnetic massive regime the amplitudes
$A_n (h, h')$, $n \ne 0$, appearing in the form-factor expansion
(\ref{tffshomouni}) of the two-point functions of spin-zero operators
show the low-$T$ asymptotic behaviour
\begin{equation} \label{universala}
     A_n (h, h') = - \biggl(\frac{\dh_1 (\i \a \g|q^2)}{\dh_1' (0|q^2)}\biggr)^2
        \frac{(-1)^\ell {\cal W}^{(2 \ell)}_\a ({\cal X}_n, {\cal Y}_n|k)}
	     {\det_{2 \ell} \{{\mathbb J} ({\cal X}_n, {\cal Y}_n, h')\}}
        \bigl(1 + {\cal O} \bigl(T^\infty + a_N (T)\bigr)\bigr) \epp
\end{equation}
This formula also holds for the `lowest excited state', $n = 1$ say,
for which ${\cal X}_1 = {\cal Y}_1 = \emptyset$ and $k = 1$, if
we agree upon the convention that, for $\ell = 0$, determinants
and products are replaced by $1$. Thus,
\begin{equation}
     A_1 (h, h') = \frac{\dh_1^2 (\i \a \g/2|q)}{\dh_2^2}
        \bigl(1 + {\cal O} \bigl(T^\infty + a_N (T)\bigr)\bigr) \epp
\end{equation}
\end{lemma}
\begin{corollary}
It follows that the derivatives of the amplitudes in 
(\ref{tffshomouni}) are
\begin{equation}
     \tst{\2} \6_{\a \g}^2 A_n (h, h')\bigr|_{\a \g = 0} =
        \frac{(-1)^\ell {\cal W}^{(2 \ell)}_0 ({\cal X}_n, {\cal Y}_n|k)}
	     {\det_{2 \ell} \{{\mathbb J} ({\cal X}_n, {\cal Y}_n, h)\}}
        \bigl(1 + {\cal O} \bigl(T^\infty + a_N (T)\bigr)\bigr) \epc
\end{equation}
where ${\cal X}_n$ and ${\cal Y}_n$ are solutions to the
higher-level Bethe Ansatz equations (\ref{hlbaes}) and $k =
\fe_n - \fe_0$. In particular,
\begin{equation}
     \tst{\2} \6_{\a \g}^2 A_1 (h, h') \bigr|_{\a \g = 0}
        = - \frac{{\dh_1'}^2}{4 \dh_2^2}
        \bigl(1 + {\cal O} \bigl(T^\infty + a_N (T)\bigr)\bigr) \epp
\end{equation}
\end{corollary}
\subsection{\boldmath Properly normalized form factors of spin-zero
operators in the low-$T$ limit -- examples}
So far we have obtained explicit expressions for all functions
occurring under the sum in (\ref{tffshomouni}) except for the
form factors $F_{X; n}^{(-)}$, $F_{Y; n}^{(+)}$ that depend on
the concrete form of the operators $X$ and $Y$. Using the JMS
theorem \cite{JMS08} these form factors can be expressed in terms
of the functions $\om_n^{(\pm)}$ and $\r_n$. There is, however,
no simple closed formula for these form factors. They rather
have to be calculated case by case (see e.g.\ \cite{GKW21,MiSm19}).
Here we only proceed with the simplest examples of spin-zero
operators that were discussed above.

The most elementary case is the case $X = Y = \s^z$. In this
case we may use (\ref{evaratzero}) in (\ref{magpm}) to obtain,
up to multiplicative corrections of the form $1 + {\cal O}\bigl(
T^\infty + a_N (T)\bigr)$,
\begin{subequations}
\begin{align}
     {\cal F}_{\s^z}^{(-)} ({\cal X}_n, {\cal Y}_n|k) & = F_{\s^z; n}^{(-)} =
	(-1)^k \re^{\i \sum_{z \in {\cal X}_n \ominus {\cal Y}_n} p(z)} - 1 \epc \\[1ex]
     {\cal F}_{\s^z}^{(+)} ({\cal X}_n, {\cal Y}_n|k) & = F_{\s^z; n}^{(+)} =
	1 - (-1)^k \re^{- \i \sum_{z \in {\cal X}_n \ominus {\cal Y}_n} p(z)} \epp
\end{align}
\end{subequations}
Note that this includes the particular case
\begin{equation}
     {\cal F}_{\s^z}^{(+)} (\emptyset, \emptyset|1) =
     - {\cal F}_{\s^z}^{(-)} (\emptyset, \emptyset|1) = 2 \epp
\end{equation}

The next simplest case is the case of the magnetic current
operator. In proper units relative to the Hamiltonian it
takes the form
\begin{equation}
     {\cal J} = - 2 \i J (\s^- \otimes \s^+ - \s^+ \otimes \s^-) \epp
\end{equation}
Hence, we may use (\ref{spincurffs}) in (\ref{evarat}) to see that
\begin{subequations}
\begin{align}
     {\cal F}_{\cal J}^{(-)} ({\cal X}_n, {\cal Y}_n|k) & = F_{{\cal J}; n}^{(-)} =
	(-1)^k \re^{\i \sum_{z \in {\cal X}_n \ominus {\cal Y}_n} p(z)}
	\sum_{z \in {\cal X}_n \ominus {\cal Y}_n} \frac{\i \e(z|0)}{2} \epc \\[1ex]
     {\cal F}_{\cal J}^{(+)} ({\cal X}_n, {\cal Y}_n|k) & = F_{{\cal J}; n}^{(+)} =
	(-1)^k \re^{- \i \sum_{z \in {\cal X}_n \ominus {\cal Y}_n} p(z)}
	\sum_{z \in {\cal X}_n \ominus {\cal Y}_n} \frac{\i \e(z|0)}{2} \epc
\end{align}
\end{subequations}
up to multiplicative corrections of the form
$1 + {\cal O}\bigl( T^\infty + a_N (T)\bigr)$. This includes
again the particular case
\begin{equation}
     {\cal F}_{\cal J}^{(+)} (\emptyset, \emptyset|1) =
     {\cal F}_{\cal J}^{(-)} (\emptyset, \emptyset|1) = 0 \epp
\end{equation}

In order to proceed e.g.\ with the energy density operator we
have to consider the low-$T$ limit of the function $\om_n^{(\pm)}$,
resp.\ $\Ps_n^{(\pm)}$, in the antiferromagnetic massive regime
which is not directly available from the literature. However,
using Appendices~A and C of \cite{DGKS16b} we can write down
the linear integral equations that determine the functions
$G_n^{(\pm)}$ in the zero-temperature limit as well as the
corresponding integral representations for $\Ps_n^{(\pm)}$.

For this purpose we introduce the functions
\begin{subequations}
\begin{align}
     \fz_0 (\la) & = \frac{k}{2} + \frac{\g \a}{2 \p \i} + F_n (\la) \epc \\[1ex]
     H^{(\pm)} (\la, \x) & = q^{\mp \a} \ctg(\la - \x - \i \g) - 
        w^{\pm 1} (\x|{\cal X}_n, {\cal Y}_n|k) \ctg(\la - \x) \epp
\end{align}
\end{subequations}
We also define ${\cal G}^{(\pm)} (\la, \x|{\cal X}_n, {\cal Y}_n|k)
= G_n^{(\pm)} (\la + \i \g/2, \x + \i \g/2)$ and well as ${\cal Z}_n^+
= {\cal X}_n$, ${\cal Z}_n^- = {\cal Y}_n$. For short we sometimes
suppress the parameter dependence and write ${\cal G}^{(\pm)} (\la, \x) =
{\cal G}^{(\pm)} (\la, \x|{\cal X}_n, {\cal Y}_n|k)$, as well as
$w(\x) = w (\x|{\cal X}_n, {\cal Y}_n|k)$. Then, up to multiplicative
corrections of the form $1 + {\cal O}\bigl( T^\infty + a_N (T)\bigr)$,
the functions ${\cal G}^{(\pm)} (\la, \x)$ satisfy the following
low-temperature linear integral equations,
\begin{multline}
     {\cal G}^{(\pm)} (\la, \x) =
        q^{\pm \a} \ctg(\la - \x - \i \g) - w^{\pm 1} (\x) \ctg(\la - \x) \\[1ex]
	\pm \sum_{z \in {\cal Z}_n^{\pm}}
	\frac{K_{\mp \a} (\la - z) {\cal G}^{(\pm)} (z, \x)}
	     {(w^{\pm 1})' (z) (\re^{2 \p \i \fz_0 (z)} - 1)}
        - \int_{- \frac{\p}{2}}^\frac{\pi}{2} \frac{\rd \m}{2 \p \i}
	  w^{\mp 1} (\m) K_{\mp \a}(\la - \m) {\cal G}^{(\pm)} (\m, \x) \epc
\end{multline}
and
\begin{multline}
     \i \Ps_n^{(\pm)} (\x_1 + \i \g/2,\x_2 + \i \g/2)
        = \int_{- \frac{\p}{2}}^\frac{\pi}{2} \frac{\rd \la}{2 \p \i}
	   w^{\mp 1} (\la) H^{(\pm)} (\la, \x_1) {\cal G}^{(\pm)} (\la, \x_2) \\[1ex]
	   \mp \sum_{z \in {\cal Z}_n^{\pm}}
	     \frac{H^{(\pm)} (z, \x_1){\cal G}^{(\pm)} (z, \x_2)}
	          {(w^{\pm 1})' (z) (\re^{2 \p \i \fz_0 (z)} - 1)}
           - \bigl(w^{\pm 1} (\x_1) - w^{\pm 1} (\x_2)\bigr) \ctg(\x_1 - \x_2) \epp
\end{multline}

So far we have not been able to solve the integral equations
for ${\cal G}^{(\pm)}$ in the general case and to obtain
fully explicit expressions for $\Ps_n^{(\pm)}$. From a practical
point of view this may not even be necessary as we have seen
in previous work \cite{DGKS16b} that very similar equations can
be efficient for the numerical evaluation of form-factor series.
We would also like to point out that the closely related functions
$\PH_n^{(\pm)}$, Eq.~(\ref{defphi}), have been obtained in closed
form starting directly from the integral equations for
${\cal G}^{(\pm)}$ \cite{BoGo10} and that a somewhat related
equation has indeed been solved explicitly in \cite{JMS11b}.
This keeps us optimistic that a more explicit characterization
of the functions $\Ps_n^{(\pm)}$ is within reach.
\enlargethispage{3ex}

\subsection{Form factor expansion of two-point functions in
the low-temperature limit}
As all excitations of the quantum transfer matrix of the XXZ
chain in the antiferromagnetic massive regime and in the low-$T$
limit are parameterised by sets ${\cal X}_n$, ${\cal Y}_n$ and
an index $k \in \{0, 1\}$, the properly normalized form factors
of any two local operators $X$, $Y$ take the form
${\cal F}_{X}^{(-)} ({\cal X}_n, {\cal Y}_n|k)$,
${\cal F}_{Y}^{(+)} ({\cal X}_n, {\cal Y}_n|k)$. As only their
product appears under the sum in (\ref{tffshomouni}) we introduce
the short-hand notation
\begin{equation}
     {\cal F}_{X, Y}^{(2 \ell)} ({\cal X}_n, {\cal Y}_n|k) = 
        {\cal F}_{X}^{(-)} ({\cal X}_n, {\cal Y}_n|k)
        {\cal F}_{Y}^{(+)} ({\cal X}_n, {\cal Y}_n|k) \epc
\end{equation}
where $\ell = \card {\cal X}_n = \card {\cal Y}_n$. Then, for
instance,
\begin{subequations}
\begin{align} \label{fszsz}
     {\cal F}_{\s^z, \s^z}^{(2 \ell)} ({\cal X}_n, {\cal Y}_n|k) & =
        - 4 \sin^2 \biggl\{ \2
	           \biggl( k \p + \mspace{-16.mu}
		          \sum_{z \in {\cal X}_n \ominus {\cal Y}_n}
			  \mspace{-16.mu} p(z)\biggr) \biggr\}
        \bigl(1 + {\cal O} \bigl(T^\infty + a_N (T)\bigr)\bigr) \epc \\[1ex] \label{fjj}
     {\cal F}_{{\cal J}, {\cal J}}^{(2 \ell)} ({\cal X}_n, {\cal Y}_n|k) & =
        - \4 \biggl(\sum_{z \in {\cal X}_n \ominus {\cal Y}_n}
	            \mspace{-16.mu} \e (z|0)\biggr)^2
        \bigl(1 + {\cal O} \bigl(T^\infty + a_N (T)\bigr)\bigr) \epp
\end{align}
\end{subequations}

In the low-$T$ limit in the antiferromagnetic massive regime the
sum in (\ref{tffshomouni}) becomes a sum over all possible triples
$({\cal X}_n, {\cal Y}_n, k)$ admitted by the higher-level Bethe
Ansatz equations (\ref{hlbaes}). Due to the factors
$\det_{2 \ell} \{{\mathbb J} ({\cal X}_n, {\cal Y}_n, h)\}$ in the
denominator the summands can be interpreted as multiple residues.
Under certain assumptions, which are discussed in detail in
Appendix D of \cite{GKKKS17}, the sum over all excitations becomes
a sum of multiple integrals. For this to become true we have to keep
the Trotter number finite at first instance and perform the Trotter
limit only after turning sums into integrals. In the last step
we can then send $T \rightarrow 0_+$. We refer to \cite{GKKKS17,BGKS21a}
for the details. The result is formulated in the following
\begin{theorem}
The dynamical two-point functions of local spin-zero operators for the
XXZ chain in its antiferromagnetic regime and in the zero-temperature
limit have the form-factor series expansion
\begin{align} \label{tffs}
     & \bigl\< X_{\llbracket 1, \ell\rrbracket} (t)
             Y_{\llbracket 1 + m, r + m\rrbracket} \bigr\> -
     \bigl\< X_{\llbracket 1, \ell\rrbracket}\bigr\>
     \bigl\<Y_{\llbracket 1, r\rrbracket} \bigr\> = \frac{{\dh_1'}^2}{4 \dh_2^2}
        {\cal F}^{(0)}_{X, Y} (\emptyset, \emptyset|1) (-1)^{m+1} \\[1ex]
        & + \sum_{\substack{\ell \in {\mathbb N}^*\\k = 0, 1}}
	  \mspace{-8.mu} \frac{(-1)^{km}}{(\ell !)^2} \mspace{-8.mu}
	     \int_{{\cal C}_h^\ell} \mspace{-3.mu}
	     \frac{\rd^\ell u}{(2\p)^\ell} \mspace{-3.mu}
	     \int_{{\cal C}_p^\ell} \mspace{-3.mu}
	     \frac{\rd^\ell v}{(2\p)^\ell} \:
	     {\cal W}_{0}^{(2 \ell)} ({\cal U}, {\cal V}|k)
	     {\cal F}_{X,Y}^{(2 \ell)} ({\cal U}, {\cal V}|k)
	     \re^{- \i \sum_{\la \in {\cal U} \ominus {\cal V}}
	               (m p(\la) - t \e (\la)) }, \notag
\end{align}
where the integration contours ${\cal C}_p$ and ${\cal C}_h$
may, for instance, be chosen as
\begin{equation} \label{defchcp}
     {\cal C}_p = \tst{[- \frac{\p}{2}, \frac{\p}{2}] + \frac{\i \g}{2}} + \i 0_+ \epc \qd
     {\cal C}_h = \tst{[- \frac{\p}{2}, \frac{\p}{2}] - \frac{\i \g}{2}}  + \i 0_+ \epp
\end{equation}
\end{theorem}

Recall that ${\cal W}_{0}^{(2 \ell)}$ has been defined in
(\ref{weightfun}) and that $p$ and $\e$ were introduced
in (\ref{ptheta4}), (\ref{dressede}). Inserting (\ref{fszsz})
or (\ref{fjj}) we obtain fully explicit series representations
for the two-point functions of the local magnetization or of
the magnetic current, respectively. The former was first
obtained from a generating function and studied numerically in
\cite{BGKS21a,BGKSS21}. The latter was published without derivation
in \cite{GKSS22}, where it was used in oder to study the optical
conductivity of the XXZ chain that is associated with the
magnetic current.

The $\ell$th term in the series comprises the contribution of
all scattering states of $\ell$ particles and $\ell$ holes to
the correlation functions. We observe a remarkable factorization
of the integrand into the universal weight ${\cal W}_0^{(2 \ell)}$
that is the same for all spin zero-operators and a part
${\cal F}_{X,Y}^{(2 \ell)}$ that does depend on the operators.
The latter consists of two factors that can be obtained from
solutions of a discrete form (\ref{rqKZ}) of a reduced qKZ type
equation and that can be expressed in terms of two functions
$\r_n$ and $\om_n^{(\pm)}$ by means of the JMS theorem \cite{JMS08}.

\section{Conclusions}
\label{sec:conclusions}
Generalizing our work \cite{GKKKS17} we have presented a form-factor
series (\ref{tffsxy}) for the dynamical two-point function at finite
temperature of arbitrary local operators in fundamental Yang-Baxter
integrable models. The summands in the series are composed of eigenvalue
ratios and form factors of a suitably defined dynamical quantum transfer
matrix. We may interpret the series as an evaluation of a lattice
path integral by Bethe Ansatz.

We have then used the series to explore the two-point functions of
spin-zero operators in the XXZ chain. For those we have suggested to
factor out from the summands a universal amplitude (\ref{ampl}), that
does not depend on the operators, and to consider properly normalized
form factors (\ref{geninhomff}). The latter have properties very
similar to generalized reduced density matrices. In particular, they
satisfy reduction relations (\ref{redurel}) and a generalized form
(\ref{rqKZ}) of the discrete rqKZ equation introduced in \cite{AuKl12}.
The reduction relations and the exchange relation (\ref{exchangerel})
are sufficient to calculate the simplest form factors, those of
the magnetization operator and of the magnetic current operator,
explicitly in terms of ratios $\r_n$ of eigenvalues of the quantum
transfer matrix. For the properly normalized form factors of
arbitrary local operators we have presented multiple-integral
representations (\ref{ffmultint}), and we have shown how they are
connected with the Fermionic basis introduced by H. Boos et al.

We have then applied our results to the two-point functions of local
operators of spin zero in the XXZ chain in its antiferromagnetic massive
regime for $T \rightarrow 0_+$. In this case we could take recourse
to our previous work \cite{BGKS21a,BGKSS21} and extract an
explicit formula (\ref{universala}) for the universal amplitude.
In combination with the explicit result (\ref{evarat}) for the
eigenvalue ratios we have obtained the formula (\ref{tffs})
for the general two-point functions in this limit in which
the only input consists of the properly normalized form factors
${\cal F}_X^{(-)}$, ${\cal F}_Y^{(+)}$. For the cases $X, Y
\in \{\s^z, {\cal J}\}$ we have provided explicit expressions as
well.

Our work leaves many directions open for future research. In first
place, we would like to derive explicit expressions for the
functions $\om_n^{(\pm)}$ in the low-$T$ limit. We are confident
that these functions are expressible in terms of known special
function for parameter values inside the antiferromagnetic massive
regime. Another important question will be the generalization of
our results to operators of arbitrary spin. Will we be able to
identify a universal weight function for non-zero spin, and, if
yes, what characterizes the family of weight functions? What
are the corresponding properly generalized form factors and how
are they related to the Fermionic basis? Last but not least, it
will be interesting to extend our work to the massless regime.
In the low-$T$ limit we may resort to our recent analysis
\cite{FGK23pp} which showed that, at least for $0 < \D < 1$,
the low-$T$ Bethe root patterns are of pure particle-hole type
and do not involve any strings. It is encouraging that a similar
fact was crucial for the derivation of the form-factor series
(\ref{tffs}) in the antiferromagnetic massive regime.

\subsection*{Acknowledgement}
The authors would like to thank H. Boos, A. Kl\"umper, F. Smirnov
and A. Wei{\ss}e for helpful discussions. FG and MM acknowledge
financial support by the German Research Council (DFG) in the
framework of the research unit FOR 2316 and through the project
Kl 645/21-1. KKK is supported by the ERC Project LDRAM: ERC-2019-ADG
Project 884584.

\clearpage

\renewcommand{\thesection}{\Alph{section}}
\renewcommand{\theequation}{\thesection.\arabic{equation}}

\begin{appendices}

\section{A proof of Lemma~\ref{lem:staginv}}
\setcounter{equation}{0}
\label{app:proofinvlemma}
\begin{proof}
(i) we note that the regularity of the $R$-matrix (\ref{reg})
implies
\begin{multline} \label{transevalinhom}
     t(\n_\ell|\siv, \nuv, h) = \\
        \begin{cases}
	   \dst{
           \Bigl[\prod_{k \in \llbracket 1, \ell - 1\rrbracket}^\dst{\curvearrowleft}
	   R_{\ell, k}^{(\s_k)} (\n_\ell, \n_k) \Bigr]
	   \th_\ell (h/T)
           \Bigl[\prod_{k \in \llbracket \ell + 1, M\rrbracket}^\dst{\curvearrowleft}
	   R_{\ell, k}^{(\s_k)} (\n_\ell, \n_k) \Bigr]}
	   & \text{if $\s_\ell = 1$}\\[3ex]
	   \dst{
           \Bigl[\prod_{k \in \llbracket \ell + 1, M\rrbracket}^\dst{\curvearrowright}
	   R_{k, \ell}^{(- \s_k)} (\n_k, \n_\ell) \Bigr]
	   \th_\ell^t (h/T)
           \Bigl[\prod_{k \in \llbracket 1, \ell - 1\rrbracket}^\dst{\curvearrowright}
	   R_{k, \ell}^{(- \s_k)} (\n_k, \n_\ell) \Bigr]}
	   & \text{if $\s_\ell = - 1$.}
	\end{cases}
\end{multline}
In the second line we have used that the trace of a matrix equals
the trace of its transpose and that $R_{0,k}^{(\s_k)} (\la, \n_k) 
\mapsto R_{k,0}^{(- \s_k)} (\n_k, \la)$ under transposition in
space $0$.

(ii) From now on we suppress the dependence of the transfer matrices
and the mono\-dromy matrices on $\nuv$ and $\siv$. Let
\begin{equation}
     T_{0; j} (\la|h) =
        \Bigl[\prod_{k \in \llbracket 1, j - 1\rrbracket}^\dst{\curvearrowleft}
	   R_{0, k}^{(\s_k)} (\la, \n_k) \Bigr]
	   \th_0 (h/T)
           \Bigl[\prod_{k \in \llbracket j, M\rrbracket}^\dst{\curvearrowleft}
	   R_{0, k}^{(\s_k)} (\la, \n_k) \Bigr] \epp
\end{equation}
It follows that
\begin{equation} \label{conventionappa}
     T_{0; 1} (\la|h) = T (\la|h)
\end{equation}
and, since $\s_j = 1$ by hypothesis,
\begin{equation} \label{inverseonesite}
     \tr_0 \bigl\{x_0 T_{0; j} (\n_j|h)\bigr\} = x_j t(\n_j|h) \epp
\end{equation}
Here we have used the regularity of the $R$-matrix and
equation~(\ref{transevalinhom}).

(iii) Let $\s_\ell = 1$. Then
\begin{equation} \label{shiftplus}
     t(\n_\ell|h) T_{0; \ell} (\la|h) = T_{0; \ell + 1} (\la|h) t(\n_\ell|h) \epp
\end{equation}
This follows from (\ref{transevalinhom}), from a multiple
use of the Yang-Baxter equation (\ref{ybe}) and from
(\ref{uonesym}). We have utilized, in particular, that
\begin{equation}
     R_{\ell, k}^{(\s_{k})} (\n_\ell, \n_{k})
        R_{0, k}^{(\s_{k})} (\la, \n_{k})
        R_{0, \ell}^{(\s_\ell)} (\la, \n_\ell) =
     R_{0, \ell}^{(\s_\ell)} (\la, \n_\ell)
        R_{0, k}^{(\s_{k})} (\la, \n_{k})
        R_{\ell, k}^{(\s_{k})} (\n_\ell, \n_{k})
\end{equation}
for $k \ne \ell$, $\s_\ell = 1$.

(iv) Let $\s_\ell = - 1$. Then a similar calculation shows that
\begin{equation} \label{shiftminus}
     T_{0; \ell} (\la|h) t(\n_\ell|h) = t(\n_\ell|h) T_{0; \ell + 1} (\la|h) \epp
\end{equation}
This time we have used the second equation (\ref{transevalinhom})
and the Yang-Baxter equation in the form
\begin{multline}
     R_{0, k}^{(\s_{k})} (\la, \n_{k})
        R_{0, \ell}^{(\s_\ell)} (\la, \n_\ell)
	R_{k, \ell}^{(- \s_{k})} (\n_k, \n_\ell) \\ =
     R_{k, \ell}^{(- \s_{k})} (\n_k, \n_\ell)
        R_{0, \ell}^{(\s_\ell)} (\la, \n_\ell)
	R_{0, k}^{(\s_{k})} (\la, \n_{k}) \epc
\end{multline}
holding for $k \ne \ell$, $\s_\ell = - 1$.

(v) If now $t(\n_\ell|h)$ is invertible we conclude with
(\ref{shiftplus}), (\ref{shiftminus}) that
\begin{equation}
     T_{0; \ell + 1} (\la|h) =
        t^{\s_\ell} (\n_\ell|h) T_{0; \ell} (\la|h) t^{- \s_\ell} (\n_\ell|h) \epp
\end{equation}
Hence,
\begin{equation} \label{shiftedmono}
     T_{0; j} (\la|h) =
        \Bigl[ \prod_{k=1}^{j-1} t^{\s_k} (\n_k|h) \Bigr] T_{0; 1} (\la|h)
        \Bigl[ \prod_{k=1}^{j-1} t^{- \s_k} (\n_k|h) \Bigr] \epc
\end{equation}
and equation (\ref{inversestagg}) follows with (\ref{conventionappa}),
(\ref{inverseonesite}) and (\ref{shiftedmono}).
\end{proof}

\section{Spin and pseudo spin}
\setcounter{equation}{0}
\label{app:pseudospin}
The $U(1)$ symmetry of the $R$-matrix (\ref{uonesym}) induces
the conservation of a `pseudo spin' by the $\siv$-staggered
transfer matrix, Section~\ref{sec:sigmastagg}.

Let
\begin{equation}
     \th^{(\s_j)} (\a) = \begin{cases}
			    \th (\a) & \text{if $\s_j = 1$} \\
			    \th^t (- \a) & \text{if $\s_j = - 1$}
                         \end{cases}
\end{equation}
and define the pseudo-spin operator
\begin{equation}
     \Si = \sum_{j=1}^M {\th_j^{(\s_j)}}' (0) \epp
\end{equation}
Then (\ref{uonesym}) implies that
\begin{equation}
     [R_{0, j}^{(\s_j)} (\la, \n_j), \th_0 (\a) \th_j^{(\s_j)} (\a)] = 0 \epp
\end{equation}
Hence,
\begin{equation} \label{pseudospingroup}
     [T (\la|\s, \n, h), \th_0 (\a) \prod_{j=1}^M \th_j^{(\s_j)} (\a)] = 0
\end{equation}
implying that
\begin{equation} \label{pseudospinsym}
     [T (\la|\s, \n, h), \th_0' (0) + \Si] = 0 \epp
\end{equation}

In case of the XXZ chain we have $\hat \ph = \2 \s^z$, and
\begin{equation} \label{pseudospinsstag}
     \Si = \2 \sum_{j=1}^M \s_j \s_j^z \epp
\end{equation}
On the one hand, if we consider (\ref{inhoms}) with $\ell = 0$, we obtain
(\ref{pseudospinop}) as a special case. In this case
(\ref{pseudospinsym}) can be spelled out as
\begin{equation} \label{adpst}
     \ad_\Si T^\a_\be (\la|h) = [\Si, T^\a_\be (\la|h)] =
        \frac{\be - \a}{2} T^\a_\be (\la|h)
\end{equation}
which immediately implies (\ref{psconservation}).
On the other hand, we have the elementary formula
\begin{equation} \label{adse}
     \ad_{\2 \s^z} e^\a_\be = \frac{\be - \a}{2} e^\a_\be \epp
\end{equation}
Comparing (\ref{adpst}) and (\ref{adse}) we obtain
Equation~(\ref{sequalsps}) of the main text. Finally,
taking into account that $\ad_\Si$ is a derivation, we
conclude that
\begin{equation}
     s_p \bigl(T^{\a_1}_{\be_1} \dots T^{\a_m}_{\be_m}\bigr)
        = \2 \sum_{j=1}^m (\be_j - \a_j) \epp
\end{equation}

\end{appendices}

%\bibliographystyle{amsplain}
%\bibliography{hub}

\providecommand{\bysame}{\leavevmode\hbox to3em{\hrulefill}\thinspace}
\providecommand{\MR}{\relax\ifhmode\unskip\space\fi MR }
\providecommand{\MRhref}[2]{%
  \href{http://www.ams.org/mathscinet-getitem?mr=#1}{#2}
}
\providecommand{\href}[2]{#2}

\end{document}